\documentclass[a4paper]{article}

\usepackage[left=0.75in,right=0.75in,top=0.75in,bottom=0.75in]{geometry}

\usepackage{graphicx}
\usepackage{xspace}

\usepackage{hyperref}

\usepackage{amsmath}
\usepackage{amsfonts}
\usepackage{amssymb}

\usepackage{tabu}
\usepackage{multirow}
\usepackage{booktabs}
\usepackage{sidecap}
\usepackage[binary-units]{siunitx}

\usepackage{float}

\usepackage[thref,amsmath]{ntheorem}

\newtheorem{lemma}{Lemma}
\newtheorem{corollary}{Corollary}
\newtheorem{definition}{Definition}

\newtheorem{proposition}{Proposition}
\newtheorem{example}{Example}
\newtheorem{observation}{Observation}
\newtheorem{problem}{Problem}

\theoremstyle{nonumberplain}
\theoremheaderfont{\itshape}
\theorembodyfont{\normalfont}
\theoremseparator{.}
\theoremsymbol{}
\newtheorem{proof}{Proof}

\providecommand{\keywords}[1]{\textbf{\textit{Keywords ---}} #1}

\usepackage{authblk}

\newbox{\myorcidaffilbox}
\sbox{\myorcidaffilbox}{\large\includegraphics[height=1.7ex]{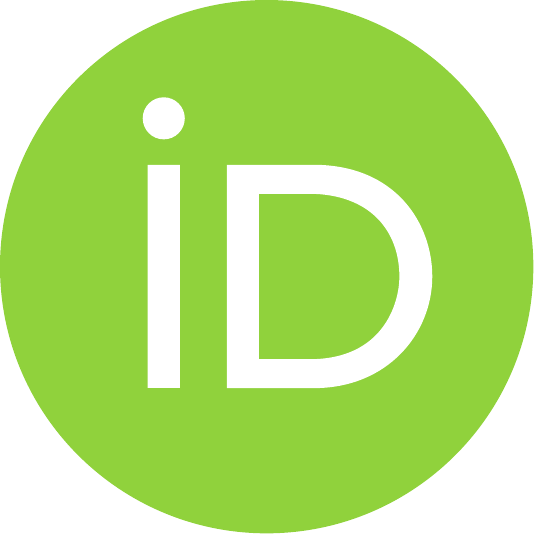}}
\newcommand{\orcid}[1]{%
	\href{https://orcid.org/#1}{\usebox{\myorcidaffilbox}}}

\usepackage{microtype}

\usepackage{tikz}
\usetikzlibrary{positioning,shapes,chains,backgrounds,arrows,chains,fit,snakes,calc,automata}

\usepackage{forest}

\usepackage[linesnumbered,vlined,ruled]{algorithm2e}

\usepackage[T1]{fontenc}
\usepackage{datetime}
\shortdate

\usepackage{comment}

\usepackage{subcaption}

\newcommand{\Oh}{\mathcal{O}}

\newcommand{\depth}{\textit{td}}
\newcommand{\stringdepth}{\textit{sd}}
\newcommand{\locus}{\textit{loc}}

\newcommand{\slink}{\textit{slink}}
\newcommand{\rev}{\textrm{rev}}
\newcommand{\Suf}{\textrm{Suf}}
\newcommand{\Pref}{\textrm{Pref}}
\newcommand{\leftmin}{\textit{Left-min}}
\newcommand{\sn}{\textit{ln}}

\newcommand{\parent}{\textit{parent}}
\newcommand{\child}{\textit{child}}

\renewcommand{\a}{\tt a}
\renewcommand{\b}{\tt b}
\renewcommand{\c}{\tt c}

\newcommand{\remove}[1]{}

\newcommand{\tikzmark}[1]{\tikz[overlay,remember picture] \node (#1) {};}
\newcommand{\DrawBox}[3][]{%
    \tikz[overlay,remember picture]{
    \draw[black,#1]
      ($(#2)+(-0.5em,2.0ex)$) rectangle
      ($(#3)+(0.75em,-0.75ex)$);}
}

\newlength\mylen
\newcolumntype{C}{>{\hfil$}p{\mylen}<{$\hfil}} %

\date{}
\begin{document}

\title{Pattern Discovery in Colored Strings}

\author[1]{Zsuzsanna Lipt\'{a}k\orcid{0000-0002-3233-0691}}
\author[2]{Simon J. Puglisi\orcid{0000-0001-7668-7636}}
\author[3]{Massimiliano Rossi\orcid{0000-0002-3012-1394}}

\affil[1]{Department of Computer Science, University of Verona, Verona, Italy}
\affil[2]{Helsinki Institute of Information Technology (HIIT), Department of Computer Science, University of Helsinki, Helsinki, Finland}
\affil[3]{Department of Computer and Information Science and Engineering, University of Florida, Gainesville, USA}
{
    \makeatletter
	\renewcommand\AB@affilsepx{, \protect\Affilfont}
	\makeatother
	\affil[1]{zsuzsanna.liptak@univr.it}
    \affil[2]{puglisi@cs.helsinki.fi}
    \affil[3]{rossi.m@ufl.edu}
}

\date{\bigskip \bigskip {\em Article published in ACM Journal of Experimental Algorithmics  (2020)} \\
doi: 10.1145/3429280}

\maketitle

\begin{abstract}
    In this paper, we consider the problem of identifying patterns of interest in colored strings. A colored string is a string where each position is assigned one of a finite set of colors. Our task is to find substrings of the colored string that always occur followed by the same color at the same distance.  The problem is motivated by applications in embedded systems verification, in particular, assertion mining. The goal there is to automatically find properties of the embedded system from the analysis of its simulation traces.

    We show that, in our setting, the number of patterns of interest is upper-bounded by $\Oh(n^2)$, where $n$ is the length of the string. We introduce a baseline algorithm, running in $\Oh(n^2)$ time, which identifies all patterns of interest satisfying certain minimality conditions, for all colors in the string. For the case where one is interested in patterns related to one color only, we also provide a second algorithm which runs in $\Oh(n^2\log n)$ time in the worst case but is faster than the baseline algorithm in practice. Both solutions use suffix trees, and the second algorithm also uses an appropriately defined priority queue, which allows us to reduce the number of computations. We performed an experimental evaluation of the proposed approaches over both synthetic and real-world datasets, and found that the second algorithm outperforms the first algorithm on all simulated data, while on the real-world data, the performance varies between a slight slowdown (on half of the datasets) and a speedup by a factor of up to 11. \footnote{A preliminary version of this paper was presented at SEA2020~\cite{LPR20}}
    
    \noindent\keywords{property testing, suffix tree, pattern mining, efficient algorithm}
\end{abstract}
\section{Introduction}\label{sec:introduction}

In recent years, embedded systems have become increasingly pervasive and are becoming fundamental components of everyday life. In line with this, embedded systems are required to perform more and more demanding tasks, and in many circumstances, peoples' lives are now dependent on the correct functioning of these devices. This, in turn, has led to an increasingly complex design process for embedded systems, where a major design task is to evaluate and check the correctness of the functionality from the early stages of the development process. This functionality checking is usually done using {\em assertions} --- logic formulae expressed in temporal logic such as Linear Temporal Logic (LTL) or Computation Tree Logic (CTL) --- that provide a way to express desirable properties of the device. Assertions are typically written by hand by the designers, and it might take months to obtain a set of assertions that is small and effective (i.e. it covers all functionalities of the device)~\cite{foster2004assertion}. In order to help designers with the verification process, methodologies and tools have been developed which automatically generate assertions from simulation traces of an implementation of the device~\cite{liu2011automatic,vasudevan2010goldmine,danese2015automatic,danese2017team}. The objective is to provide a small set of assertions that cover all behaviors of the device, in order to extend the basic manually-defined set of assertions.

\medskip

A simulation trace can be viewed as a table that records, for every simulation instant $T$, the value assumed by the input and output ports of the device. Figure~\ref{tab:example of simulation trace} shows an example of a simulation trace of a device with three input ports ${\cal I} = \{i_1, i_2, i_3\}$ and two output ports ${\cal O} = \{o_1, o_2\}$. An assertion is a logic formula expressed in temporal logic that must remain true in the whole trace. The simplest assertions involve only conditions occurring at the same simulation instant. In the simulation trace in Figure~\ref{tab:example of simulation trace}, from the solid and dashed shaded boxes, we can assert that each time we have $i_1 = 1$, $i_2 = 0$, and $i_3 = 1$, then $o_1 = 1$ and $o_2 = 1$. On the other hand, we cannot assert that each time we have $i_1 = 1$, $i_2 = 1$, and $i_3 = 0$, then $o_1 = 1$ and $o_2 = 1$, because there is a counterexample in the simulation trace, namely at instant $T= 9$, where $o_1 = 0$ and $o_2 = 0$. Note that the assertions do not need to contain all input and output variables, e.g.\ we can assert that $i_1=0$ and $i_3=0$ implies $o_2=0$.

Among all possible types of assertions that can be expressed in temporal logic, an interesting one is given by chains of {\em next:} sequences of consecutive input values that, when provided to the device, uniquely determine their output, after a certain number of simulation instants. For example, in the simulation trace in Figure~\ref{tab:example of simulation trace}, we can assert that each time we have, for $(i_1,i_2,i_3)$, the values $(0,1,0),(1,1,0),(0,1,0)$ in consecutive simulation instants, then, three instants later, we will see $o_1 = 1$ and $o_2 = 0$.

\begin{figure}[h]
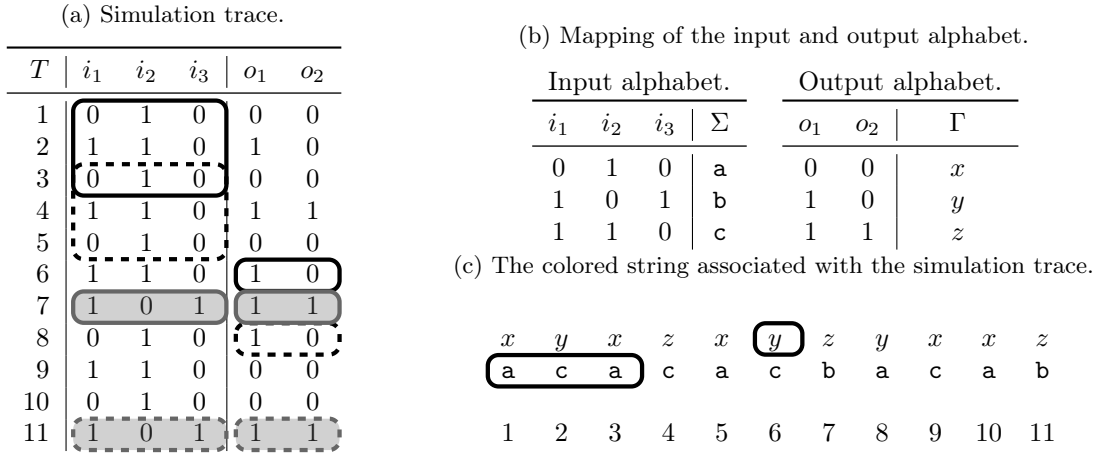

  \centering
  \begin{subfigure}[b]{.35\textwidth}
    \centering
    \caption{Simulation trace.\label{tab:example of simulation trace}}
    \begin{tabular}{r|ccc|cc}
        \toprule
        $T$ & $i_1$ & $i_2$ & $i_3$ & $o_1$ & $o_2$ \\
        \midrule
        1 & \tikzmark{top left 1}0 & 1 & 0 & 0 & 0 \\
        2 & 1 & 1 & 0 & 1 & 0 \\
        3 & \tikzmark{top left 3}0 & 1 & 0\tikzmark{bottom right 1}& 0 & 0 \\
        4 & 1 & 1 & 0 & 1 & 1 \\
        5 & 0 & 1 & 0\tikzmark{bottom right 3} & 0 & 0 \\
        6 & 1 & 1 & 0 & \tikzmark{top left 2}1 & 0\tikzmark{bottom right 2}\\
        7 & \tikzmark{top left 5}1 & 0 & 1\tikzmark{bottom right 5} & \tikzmark{top left 6}1 & 1\tikzmark{bottom right 6} \\
        8 & 0 & 1 & 0 & \tikzmark{top left 4}1 & 0\tikzmark{bottom right 4}\\
        9 & 1 & 1 & 0 & 0 & 0 \\
        10 & 0 & 1 & 0 & 0 & 0 \\
        11 & \tikzmark{top left 7}1 & 0 & 1\tikzmark{bottom right 7} & \tikzmark{top left 8}1 & 1\tikzmark{bottom right 8} \\
    \end{tabular}

      \DrawBox[ultra thick, rounded corners,  black]{top left 1}{bottom right 1}
      \DrawBox[ultra thick, rounded corners,  black]{top left 2}{bottom right 2}
      \DrawBox[ultra thick, dashed,  rounded corners,  black]{top left 3}{bottom right 3}
      \DrawBox[ultra thick, dashed,  rounded corners,  black]{top left 4}{bottom right 4}
      \DrawBox[ultra thick, fill=black!60, fill opacity=0.3,  rounded corners,  black!60]{top left 5}{bottom right 5}
      \DrawBox[ultra thick, fill=black!60, fill opacity=0.3,  rounded corners,  black!60]{top left 6}{bottom right 6}
      \DrawBox[ultra thick, dashed, fill=black!60, fill opacity=0.3,  rounded corners,  black!60]{top left 7}{bottom right 7}
      \DrawBox[ultra thick, dashed, fill=black!60, fill opacity=0.3,  rounded corners,  black!60]{top left 8}{bottom right 8}
  \end{subfigure}\hspace{0.5em}%
  \begin{subfigure}[b]{.55\textwidth}
    \centering
    \begin{subfigure}[b]{\textwidth}
      \centering
      \caption{Mapping of the input and output alphabet.\label{fig:mapping tables}}
      \begin{tabular}{ccc|c}
        \multicolumn{4}{c}{Input alphabet.}\\
        \toprule
        $i_1$ & $i_2$ & $i_3$ & $\Sigma$\\
        \midrule
        0 & 1 & 0 & $\a$\\
        1 & 0 & 1 & $\b$\\
        1 & 1 & 0 & $\c$\\

      \end{tabular}\hspace{1em}
      \begin{tabular}{cc|c}
        \multicolumn{3}{c}{Output alphabet.}\\
        \toprule
        $o_1$ & $o_2$ & $\Gamma$\\
        \midrule
        0 & 0 & $x$\\
        1 & 0 & $y$\\
        1 & 1 & $z$\\

      \end{tabular}
    \end{subfigure}\vfill%
    \begin{subfigure}[b]{\textwidth}
      \centering
      \caption{The colored string associated with the simulation trace.\label{fig:resulting colored string}}
      \begin{center}
      	\settowidth\mylen{$10$}
        $
        \begin{array}{*{11}{C}}
        x & y & x & z & x & \tikzmark{top left 12}y\tikzmark{bottom right 12} & \tikzmark{top left 16}z\tikzmark{bottom right 16} & \tikzmark{top left 14}y\tikzmark{bottom right 14} & x & x & \tikzmark{top left 18}z\tikzmark{bottom right 18}\\
        \tikzmark{top left 11}\a & \c & \tikzmark{top left 13}\a\tikzmark{bottom right 11} & \c & \a\tikzmark{bottom right 13} & \c & \tikzmark{top left 15}\b\tikzmark{bottom right 15} & \a & \c & \a & \tikzmark{top left 17}\b\tikzmark{bottom right 17} \\
      \\
        1 & 2 & 3 & 4 & 5 & 6 & 7 & 8 & 9 & 10 & 11
        \end{array}
        $

        \DrawBox[ultra thick, rounded corners,  black]{top left 11}{bottom right 11}
        \DrawBox[ultra thick, rounded corners,  black]{top left 12}{bottom right 12}
      \end{center}
    \end{subfigure}%
  \end{subfigure}%
  \caption{Example of a simulation trace of a device with input ports ${\cal I} = \{i_1, i_2, i_3\}$, and output ports ${\cal O} = \{o_1, o_2\}$. The mapping of the input and output values of the trace into the input and output alphabet respectively. The colored string associated to the simulation trace, after the mapping. The solid and dashed shaded and non-shaded boxed values in the simulation trace highlight that every time we see the sequence of input values, then we have the corresponding output value. The solid non-shaded boxed characters in the colored string are the mapping of the corresponding solid non-shaded boxed values in the simulation trace.}
\end{figure}

\medskip

We model simulation traces with {\em colored strings}. A colored string is a string over an alphabet $\Sigma$, where each position is additionally assigned a color from an alphabet $\Gamma$. We will set $\Sigma$ as the set of tuples of possible values for the input ports $i_1,\ldots,i_k$ and $\Gamma$ as that of the output traces $o_1,\ldots,o_r$. The objective then is to identify patterns in the string whose occurrence is always followed by the same color at some given distance.

\subsection{Related Work}

Pattern mining was originally motivated by the need to discover frequent itemsets and association rules in basket data, i.e.\ items that were frequently bought together in a retail store. The seminal {\em Apriori algorithm}~\cite{agrawal1994fast} can discover that type of pattern and has become very popular (with many extensions and variations) due to its wide applicability in other data-intensive domains. Time relationship, e.g., between entries of the database in which the basket data are stored, were later considered in so-called {\em sequential pattern mining}~\cite{agrawal1995mining}. 

In sequential pattern mining, {\em episodes} are partially ordered sequences of events that appear close to each other in the sequence~\cite{mannila1997discovery}. 
Given episodes of the sequence, it is possible to build {\em episode rules} that establish antecedent-consequent relations among episodes. Sequential pattern mining has many applications (see,e.g.,~\cite{cho2008tree,laxman2008stream,fahed2018deer}) and has been surveyed extensively~\cite{pei2007constraint, mabroukeh2010taxonomy, fournier2017survey}.

Unfortunately, the above setting is not applicable to our problem, since here time is given only in a relative sense, i.e., whether an event happens before (or after) another event, while we need to count exactly the instants that occurs between the two events.

In the {\em string mining problem}~\cite{fischer2006optimal,fischer2005fast,fischer2008space,dhaliwal2010practical,valimaki2012distributed}, one aims to discover strings that appear as a substring in more than $\omega$ strings in a collection, where $\omega$ is a user-defined parameter called {\em support} of the string.
This can be also used to find strings that discriminate between two collections, i.e., strings that are frequent in one collection and not frequent in the other. These strings are called {\em emerging strings} and find important applications in data mining~\cite{Raza2019}, knowledge discovery in databases~\cite{chan2003mining} and in bioinformatics~\cite{birzele2006new}. In the field of knowledge discovery on databases, the problem has been extended to mining frequent subsequences~\cite{iwanuma2005extracting} and distinguishing subsequence patterns with gap constraints~\cite{ji2007mining,pang2017mining,wang2014efficient,wu2013pmbc}.

In~\cite{Hui92} Hui proposed a solution for the {\em color set size problem}. Here, given a tree and a coloring of its leaves, the objective is to find, for all internal nodes of the tree the number of distinct colors in the leaves of its subtree. In~\cite{Hui92}, the color set size problem is applied to several string matching and string mining problems, e.g., given a collection of $m$ strings, find the longest pattern which appears in at least $1 \leq k \leq m$ strings. Note that if the tree of the color set size problem is the suffix tree of a string $s$, then $s$ with the coloring of its suffixes can be seen as a colored string. In spite of this similarity, both, the problems that we solve, and the approaches we use, are different.

In assertion mining, the two existing tools, {\em GoldMine}~\cite{vasudevan2010goldmine} and {\em A-Team}~\cite{danese2017team}, are based on data mining algorithms. In particular, {\em GoldMine}~\cite{vasudevan2010goldmine}
extracts assertions that predicate only on one instant of the simulation trace---i.e. they do not involve any notion of time---, using decision tree based mining or association mining~\cite{agrawal1994fast}. Furthermore, using static analysis techniques together with sequential pattern mining, it extracts temporal assertions.
The tool A-Team~\cite{danese2017team}, requires the user to provide the template of the temporal assertions that they want to extract. For example, in order to extract the properties of our example in Fig.~\ref{tab:example of simulation trace}, one needs to provide a template stating that we want a property of the form: ``a property $p_1$, at the next simulation instant a property $p_2$, at the next simulation instant a property $p_3$, then after three simulation instants a property $p_4$''.
Given a set of templates, the software, using an Apriori algorithm, extracts propositions (logic formulae containing the logical connectives $\neg$, $\vee$, and $\wedge$) from the trace. Once the propositions have been extracted, the tool generates the assertion by instantiating the extracted propositions in the templates, using a decision-tree-based algorithm to find formulas that fit in the template and are verified in the simulation trace, i.e.\ if the trace contains no counterexample.

\subsection{Our Contribution}

In this work we introduce colored strings, and propose and analyze two pattern discovery problems on colored strings which correspond to a useful simplification of pattern mining w.r.t.\ assertion mining. In both problems, we are given a colored string as input. Given a colored string and a color as input, in the first problem, we must find all minimal substrings that occur followed always at the same distance by the given color. In the second problem, the color is not fixed, thus we want to find all minimal substrings which occur followed always at the same distance by the same color. We define these problems formally in Section~\ref{sec:basics}.

\medskip

Although these problems are simpler than the original assertion mining problem, the solution to our problem contains all the information, possibly filtered, to recover the desired set of minimal assertions  in a second stage.  For example, let us assume that the device that produced the simulation trace in Figure~\ref{tab:example of simulation trace} has the following behavior: every time that $i_1 = 0$, at the next instant $i_1 = 1$, and at the next instant $i_1 = 0$, then after three instants $o_1 = 1$ and $o_2 = 0$.
A solution to our problem will include all patterns of length $3$ for which $i_1$ takes consecutive values $0,1,0$, while $i_2$ and $i_3$ have arbitrary values, since all of these will result in $o_1 = 1$ and $o_2 = 0$ three instances later.

\medskip

We first upper bound the number of minimal patterns by $\Oh(n^2)$. We then propose two algorithms which find all minimal patterns, when only one color is of interest ({\tt base}), and when one is interested in all colors ({\tt base-all}). Both of these algorithms use the suffix tree of the reverse string as underlying data structure. We note that since this is a pattern mining problem, every efficient algorithm for the problem will necessarily use a dedicated string data structure (or index), such as a suffix tree, since all occurrences of substrings have to be considered concurrently.

Then we show that, in the case of one color, the first algorithm can be improved. The new algorithm, referred to as {\tt skipping}, also uses the suffix tree as its underlying data structure, together with an appropriately defined priority queue. This allows us to reduce the number of computations in practice, even though the theoretical running time of the new algorithm is worse, namely $\Oh(n^2 \log n)$. We provide an experimental evaluation of the proposed approaches. Finally, we consider (practically motivated) restrictions on the patterns and show that under these restrictions performance of the {\tt skipping} algorithm is further improved. 

\medskip

The paper is structured as follows. In Section~\ref{sec:basics} we fix definitions and notations and give the problem statements. In Section~\ref{sec:baseline} we present baseline algorithms, {\tt base} and {\tt base-all}, which solve these problems. In Section~\ref{sec:skipping} we present the modified algorithm {\tt skipping}, which solves the pattern discovery problem for only one color. In Section~\ref{sec:output} we introduce real-world data oriented restrictions on the output. In Section~\ref{sec:experiments} we present an experimental evaluation of the proposed approaches. Conclusions and future work can be found in Section~\ref{sec:conclusion}.

\section{Basics}\label{sec:basics}

Let $\Sigma$ be a finite ordered set. We refer to $\Sigma$ as {\em alphabet} and to its elements as {\em characters}.
A {\em string over $\Sigma$} is a finite sequence of characters $S = S[1,n]$, where $|S|=n$ is the {\em length} of string $S$. We denote by $\varepsilon$ the {\em empty string}, the unique string of length $0$. Note that we number strings starting from 1, and we use the array-notation for strings: we denote the $i$'th character of $S$ by $S[i]$ and use $S[i,j]$ to refer to the string $S[i]\cdots S[j]$, if $i\leq j$, while $S[i,j] = \varepsilon$ if $i>j$. Given string $S = S[1,n]$, the {\em reverse string} is the string $S^{\rev}=S[n]S[n-1]\cdots S[1]$. For string $S$ and $1\leq i \leq n$, $\Pref_i(S) = S[1,i]$ is called the {\em $i$'th prefix} of $S$, and $\Suf_i(S) = S[i,n]$ is called the {\em $i$'th suffix of $S$}. A {\em substring} of a string $S$ is a string $T$ for which there exist $i,j$ s.t.\ $T = S[i,j]$; in this case the position $i$ is referred to as an {\em occurrence} of $T$ in $S$. A substring $T$ of $S$ is called {\em proper} if $T\neq S$. When $S$ is clear from the context, then we may refer to $T$ simply as a {\em substring}.

\subsection{Colored strings}\label{sec-colored-strings}

Given two finite sets $\Sigma$ (the alphabet) and $\Gamma$ (the colors), a {\em colored string} over $(\Sigma,\Gamma)$ is a string $S=S[1,n]$ over $\Sigma$ together with a coloring function $f_S: \{1,\ldots,n\} \to \Gamma$. We denote by $\sigma=|\Sigma|$ and $\gamma=|\Gamma|$ the number of characters resp.\ of colors. Given a colored string $S$ of length $n$, its reverse is denoted $S^{\rev}$, and its coloring function $f_{S^{\rev}}$ is defined by $f_{S^{\rev}}(i) = f_S(n-i+1)$, for $i=1,\ldots, n$. When $S$ is clear from the context, we write $f$ for $f_S$ and $f^{\rev}$ for $f_{S^{\rev}}$.

We are interested in those substrings which are always followed by a given color $y$, at a given distance $d$. Look at the following example.

\begin{example}\label{ex:ex1} Let $S = acacacbacab$, with colors $xyxzxyzyxxz$:

 \medskip

 \begin{center}
   	\settowidth\mylen{$10$}
    $
    \begin{array}{*{11}{C}}
    x & y & x & z & x & y & z & y & x & x & z\\
    \a & \c & \a & \c & \a & \c & \b & \a & \c & \a & \b \\
    \\
    1 & 2 & 3 & 4 & 5 & 6 & 7 & 8 & 9 & 10 & 11
    \end{array}
    $
  \end{center}

  \medskip

  The substring ${\tt aca}$ occurs 3 times in $S$, at positions $1,3,$ and $8$. In positions 1 and 3 it is followed by a $y$ at distance $3$, while at position 8, the corresponding position is beyond the end of the string.
\end{example}

This leads to the following definitions:

\begin{definition}[$y$-good, $y$-unique, minimal]
 	\label{def-y-unique}
 	Let $S$ be a colored string over $(\Sigma,\Gamma)$, $y\in \Gamma$ a color, $d \leq n$ a non-negative integer, and $T=T[1,m]$ a substring of $S$.

 	\begin{enumerate}
 		\item An occurrence $i$ of $T$ is called {\em $y$-good with delay $d$} (or {\em $(y,d)$-good}) if $f(i+m-1+d) = y$.
 		\item $T$ is called {\em $y$-unique with delay $d$} (or {\em $(y,d)$-unique}) if for every occurrence $i$ of $T$, $i$ is $(y,d)$-good or $i+m-1+d> n$.
 		\item $T$ is called {\em minimally $(y,d)$-unique} if there exists no proper substring $U$ of $T$ which is $y$-unique with delay $d'$, for some $d'$ s.t.\ $U = T[i,j]$ and $d' = d+ |T|-j$.
 	\end{enumerate}

\end{definition}

In the Example~\ref{ex:ex1}, the occurrence of {\tt aca} in position $1$ is $(y,3)$- and $(y,5)$-good, that in position $3$ is $(y,1)$- and $(y,3)$-good, while that in position $8$ is not $(y,d)$-good for any $d$. Therefore, the substring $T={\tt aca}$ is a $(y,3)$-unique substring of $S$, since every occurrence $i$ of ${\tt aca}$ is either $(y,3)$-good (position\ 1 and 3) or $i+m-1+d >  n$ (position 8). However, ${\tt aca}$ is not minimal, since its substring ${\tt ca}$ is also $(y,3)$-unique (where $d'=d$, since {\tt ca} is a suffix of ${\tt aca}$).\\

The introduction of minimally $(y,d)$-unique substrings serves to restrict the output size.
Let $T = aXb$ be $(y,d)$-unique, with $a,b \in \Sigma$ and $X\in \Sigma^*$. We call $T$ {\em left-minimal} if $Xb$ is not $(y,d)$-unique, and {\em right-minimal} if $aX$ is not $(y,d+1)$-unique. We make the following simple observations about $(y,d)$-unique substrings. (Note that {\em 2} is a special case of {\em 3}.)

 \begin{observation}\label{obs:1} Let $S\in \Sigma^*$ and let $T$ be a $(y,d)$-unique substring of $S$.
  \begin{enumerate}
    \item $T$ is minimal if and only if it is left- and right-minimal.
    \item If $T$ is a suffix of $T'$, then $T'$ is also $(y,d)$-unique.
    \item If $T' = UTV$ is a superstring of $T$ such that $|V|\leq d$, then $T'$ is $(y,d-|V|)$-unique.
  \end{enumerate}
 \end{observation}

We are now ready to formally state the problems treated in this paper.

\begin{problem}[Pattern Discovery Problem]
  Given a colored string $S$ and a color $y$, report all pairs $(T,d)$ such that $T$ is a minimally $(y,d)$-unique substring of $S$.
\end{problem}

\begin{problem}[Unrestricted-Output Pattern Discovery Problem]\label{prob:unrestricted}
  Given a colored string $S$, report all triples $(T,y,d)$ such that $T$ is a minimally $(y,d)$-unique substring of $S$.
\end{problem}

We next give an upper bound on the number of minimally $(y,d)$-unique substrings.

\begin{lemma}\label{lemma:UB}
Given string $S$ of length $n$, the number of minimally $(y,d)$-unique substrings of $S$, over all $y\in \Gamma$ and $d=0,\ldots,n$, is $\Oh(n^2)$.
\end{lemma}

\begin{proof}
Note that, given a position $j$ and a delay $d$, every substring occurrence ending in $j$ is $(f_S(j+d),d)$-good. Therefore, for a substring $U$ with an occurrence ending in position $j$, and for fixed $d$, it holds that, if $U$ is $(y,d)$-unique for some $y$, then $y=f_S(j+d)$. Moreover, it follows from Observation~\ref{obs:1} that, given $y,d$, and $j$, at most one minimally $(y,d)$-unique substring can end at position $j$. Altogether we have that the number of minimally $(y,d)$-unique substrings is $\Oh(n^2)$, over all $y$ and $d$.
\end{proof}

\subsection{Suffix trees and suffix arrays}

Let $S$ be a string over $\Sigma$ and $\$$ a new character not belonging to $\Sigma$. We denote by ${\cal T}(S)$ the {\em suffix tree} of $S\$$, i.e.\ the compact trie of the suffixes of $S\$$. For a general introduction to suffix trees, see, e.g.,~\cite{g1997,s2003,MBCT2015}.
Here we recall some basic facts.

The suffix tree ${\cal T}(S)$ is a rooted tree in which all internal nodes are branching. Each edge is labeled with a non-empty substring of $S$ such that the labels of any two outgoing edges from the same node start with a different character. Edge labels are stored in form of two pointers $[i,j]$ into the string with the property that $S[i,j]$ equals the label of the edge. If $|S|=n$, then ${\cal T}(S)$ has exactly $n+1$ leaves, each labeled by a position from $\{1, \ldots, n+1\}$, denoted $\sn(v)$ (for {\em leaf number}). For a node $v$ in ${\cal T}(S)$, we denote by $L(v)$ the concatenation of the edge labels on the path from the root to node $v$. The string $L(v)$ is sometimes referred to as the substring represented by node $v$. If $v$ is a leaf with $\sn(v)=i$, then $L(v)$ is equal to the $i$'th suffix of $S\$$, $\Suf_i(S\$)$. For a node $v$, we denote by $\depth(v)$ its {\em treedepth}, the number of edges on the path from the root to $v$, and by $\stringdepth(v)=|L(v)|$ its {\em stringdepth}, the length of the string represented by $v$. %
Given node $v$ not equal to the root, $\parent(v)$ is the node which is next on the path from $v$ to the root. Given a node $v$ which is not a leaf and a character $c\in \Sigma$, $\child(v,c)$ returns the unique node $u$ with parent $v$ such that the label of the edge $(v,u)$ starts with character $c$, or the empty pointer if no such node exists.

Given a node $u$ with parent $v$, a {\em locus} is a pair $(u,t)$ s.t.\ $\stringdepth(v) < t \leq \stringdepth(u)$. Let $[i,j]$ be the label of edge $(v,u)$ and $k = t-\stringdepth(v)$. We define $L(u,t)$ as the string $L(v)\cdot S[i,i+k-1]$, the substring represented by locus $(u,t)$. Note that if $t = \stringdepth(u)$, then $L(u,t) = L(u)$. It is an important property of suffix trees that there is a one-to-one correspondence between loci of ${\cal T}(S)$ and substrings of $S\$$. This allows us to define, for a substring $T$ of $S$ (which is also a substring of $S\$$), the {\em locus of $T$}, $\locus(T)=\locus(T,{\cal T}(S))$ as the unique locus $(u,t)$ in ${\cal T}(S)$ with the property that $L(u,t) = T$. Given a substring $T$ of $S$ with locus $\locus(T) = (u,t)$, the set of occurrences of $T$ is given by the set $\{\sn(v) \mid v \text{ is leaf in the subtree rooted in } u\}$.

Let $u$ be a node and $L(u) = aT$, where $a\in \Sigma$ and $T\in \Sigma^*$. The {\em suffix link} of $u$ is defined as $\slink(u)=\locus(T)$. It can be shown that for any node $u$, $\slink(u)$ is a node of ${\cal T}(S)$ (rather than just a locus). Suffix links can also be defined for loci: for locus $(u,t)$ with $L(u,t) = aT$, define $\slink(u,t) = \locus(T)$; these are also called {\em implicit suffix links}. Suffix links are often represented by directed edges, see Figure~\ref{fig:ex1}.

Given a suffix tree ${\cal T}(S)$ with $k$ nodes, and a node $u$ of ${\cal T}(S)$, let $r$ be the rank of the node $u$ in the breath-first search traversal of the tree. We define the reverse index BFS of $u$ as $iBFS(u) = k - r$. Refer to Figure~\ref{fig:example h function} for an example of the reverse index BFS values.

\medskip

Given the string $S$ of length $n$, we denote by $SA_S[1,n+1]$ the {\em suffix array} of $S\$$. We refer the reader to, e.g.,~\cite{MBCT2015}, for a general introduction to suffix arrays.

The suffix array $SA_S[1,n+1]$ of a string $S\$$ is a permutation of $\{1,\ldots,n+1\}$ such that $SA_S[i] = j$ if and only if $S[j,n]\$$ is $i$-th suffix in the lexicographically ordered list of suffixes of $S\$$. The suffix array $SA_S$ and the suffix tree ${\cal T}(S)$ are deeply related.
We can obtain the $SA_S$ by listing the leaves of the suffix tree ${\cal T}(S)$ from left to right, assuming that the children are ordered according to the first characters of their edge labels. In particular, for an inner node $u$, the leaves in the subtree rooted in $u$ yield an interval of the suffix array $SA_S[i,j]$ such that $\{\sn(v) \mid v \text{ is leaf in the subtree rooted in } u\} = \{SA_S[k] \mid i \leq k \leq j\}$.

\subsection{Maximum-oriented indexed priority queue}

A maximum-oriented indexed priority queue~\cite[Sec. 2.4]{SK2011} denoted by $IPQ$, is a data structure that collects a set of $m$ items with keys ${k_1,\ldots,k_m}$, and provides the following operations:
\begin{itemize}
\item {\tt insert($i$,$k$):} insert the element at index $i$ with key $k_i = k$.
\item {\tt promote($i$,$k$):} increase the value of the key $k_i$, associated with $i$, to $k\geq k_i$.
\item {\tt demote($i$,$k$):} decrease the value of the key $k_i$, associated with $i$, to $k\leq k_i$.
\item $(i,k)\gets${\tt max():} return the index $i$ and the value $k$ of the item with maximum key $k_i$; if two items have the same key value, we report the item with larger index.
\item $k\gets${\tt keyOf($i$):} return the value of the key $k_i$ associated with index $i$.
\item $b\gets${\tt isEmpty():} return $true$ if the $IPQ$ is empty and $false$ otherwise.
\item {\tt delete($i$):} remove the element at index $i$ from the $IPQ$.
\end{itemize}

The operations {\tt insert}, {\tt promote}, {\tt demote} and {\tt delete} run in $\Oh(\log(m))$ time, while the operations {\tt max}, {\tt keyOf} and {\tt isEmpty} are performed in $\Oh(1)$ time.

For our purposes, we also require a function $b\gets${\tt allNegative()} that returns $true$ if all key values are negative, and $false$ otherwise.

We use the $IPQ$ to store keys associated to nodes $u$ of a suffix tree ${\cal T}(S)$ using $iBFS(u)$ as index.
For ease of presentation, in slight abuse of notation, we will use $u$ and $iBFS(u)$ interchangeably.

\subsection{Rank, select, and range maximum query}

A {\em bitvector} $B[1,n]$ of length $n$ is an array of $n$ bits. For all $1 \leq i \leq n$ and $b\in\{0,1\}$, we define ${\tt rank}_b(B,i)$ as the number of occurrences of $b$ in $B[1,i]$, and ${\tt select}_b(B,i)$ as the index of the $i$-th occurrence of the symbol $b$ in $B$. If $i > {\tt rank}_b(B,n)$ then ${\tt select}_b(B,i) = n+1$. Furthermore, we set ${\tt select}_b(B,0) = 0$. For both ${\tt rank}$ and ${\tt select}$ operations, if $b$ is omitted we assume $b=1$.
Given a bitvector $B$, ${\tt rank}$ and ${\tt select}$ operations can be supported in $\Oh(1)$ time using $o(n)$ bits of extra space~\cite{clark1997compact}.

\medskip

For an array  $A[1,n]$ of $n$ integers and $1 \leq i \leq j \leq n$, a {\em range maximum query} ${\tt rMq}_A(i,j)$ returns the position of the maximum element of $A[i,j]$. This answer can be provided in $\Oh(1)$ time using $2n + o(n)$ bits of space~\cite{FH11}.

Given $A[1,n]$ and the range maximum query data structure for $A$, we can compute the position of the second greatest element of $A[i,j]$ in $\Oh(1)$ time. In particular,
let $a = {\tt rMq}_A[i,j]$, we have three cases: (i) if $a = i$, then the position of the second greatest element of $A[i,j]$ is $c = {\tt rMq}_A[a+1,j]$; (ii) if $a = j$, then the position of the second greatest element of $A[i,j]$ is $b = {\tt rMq}_A[i,a-1]$; (iii) otherwise, let $b = {\tt rMq}_A[i,a-1]$, and $c = {\tt rMq}_A[a+1,j]$. The position of the second greatest element of $A[i,j]$ is $b$ if $A[b]\geq A[c]$, otherwise it is $c$, since $A[c] > A[b]$.

\section{A pattern discovery algorithm for colored strings using the suffix tree}\label{sec:baseline}

Our main tool will be the suffix tree of the reverse string, ${\cal T} = {\cal T}(S^{\rev})$.  Note that loci in ${\cal T}$ correspond to ending positions of substrings of $S$ in the following sense. Given a locus $(u,t)$ of ${\cal T}$, let $U = L(u,t)^{\rev}$. Then $U$ is a substring of $S$, and its occurrences are exactly the positions $i-|U|+1$, where $i = n-\sn(v)+1$ for some leaf $v$ in the subtree rooted in $u$. In the next lemma we show how to identify $(y,d)$-unique substrings of $S$ with ${\cal T}$,  the suffix tree of $S^{\rev}$.

\begin{lemma}\label{lemma:1}
  Let $U$ be a substring of $S$, ${\cal T} = {\cal T}(S^{\rev})$, and $(u,t) = \locus(U^{\rev},{\cal T})$. Then $U$ is $y$-unique with delay $d$ in $S$ if and only if for all leaves $v$ in the subtree rooted in $u$, $S^{\rev}[\sn(v)-d]$ is colored $y$ under $f^{\rev}$. In particular, $U$ is $y$-unique with delay $0$ in $S$ if and only if all leaves in the subtree rooted in $u$ are colored $y$ under $f^{\rev}$.
\end{lemma}

\begin{proof}
  It is easy to see that position $i-|U|+1$ is a $y$-good occurrence of $U$ in $S$ with delay $0$ if and only if $U^{\rev}$ is a prefix of $\Suf_{n-i+1}(S^{\rev})$ and $f^{\rev}(n-i+1) = y$. By the properties of the suffix tree, all occurrences of $U^{\rev}$ correspond to the leaves of the subtree rooted in $u$, where $(u,t) = \locus(U^{\rev},{\cal T})$. Thus, $U$ is $(y,0)$-unique if and only if all of its occurrences are $(y,0)$-good, which is the case if and only if all leaves of the subtree rooted in $u$ are colored $y$ under $f^{\rev}$. More generally, position $i-|U|+1$ is a $y$-good occurrence of $U$ in $S$ with delay $d$ if and only if $\Suf_{n-i+1}(S^{\rev})$ is prefixed by $U^{\rev}$ and $f^{\rev}(n-i+1-d) = y$. Thus $U$ is $(y,d)$-unique if and only if for all leaves $v$ in the subtree rooted in $u$, $S^{\rev}[\sn(v)-d]$ is colored $y$ under $f^{\rev}$.
\end{proof}

In the following, we will refer to a {\em node} $u$ of ${\cal T}$ as {\em $(y,d)$-unique} if $L(u)^{\rev}$ is a $(y,d)$-unique substring of $S$. We can now state the following corollary:

\begin{corollary}\label{corollary: all children are (y,d)-unique}
  Let $U$ be a substring of $S$, ${\cal T} = {\cal T}(S^{\rev})$, and $(u,t) = \locus(U^{\rev},{\cal T})$ such that $u$ is an inner node of ${\cal T}(S)$. Then $U$ is $(y,d)$-unique in $S$ if and only if all children of $u$ are $(y,d)$-unique.
\end{corollary}

\subsection{Finding all $(y,d)$-unique substrings}\label{sec:first_algo}

Our first algorithm {\sc Algo1} uses the suffix tree ${\cal T}$ of the reverse string to identify all $(y,d)$-unique substrings of $S$, not only the minimal ones, for fixed $y$ and $d$. It marks the $(y,d)$-unique nodes of ${\cal T}$ in a postorder traversal of the tree. Note that if $i>n-d$, then position $i$ cannot be $(y,d)$-good, simply because the position in which we would expect a $y$ lies beyond the end of string $S$. In correspondence with the definition of $(y,d)$-unique substrings (Definition~\ref{def-y-unique}), we will treat such positions as if they were $(y,d)$-good.

The function $g(u):V({\cal T}) \to \{0,1\}$ is defined as follows:

\medskip

\begin{itemize}
	\item   for a leaf $u$ with leaf number $\sn(u) = i$:
	\[ g(u) = \begin{cases}
	1 & \text{ if either $i \leq d$ or $f^{\rev}(i-d)=y$,}\\
	0 & \text{ otherwise},
	\end{cases}\]
	\item for an inner node $u$:
	\[ g(u) =
	\begin{cases}
	y & \text{ if $g(v)=1$ for all children $v$ of $u$,}\\
	0 & \text{ otherwise.}
	\end{cases} \]
\end{itemize}

\medskip

The algorithm computes $g(u)$ for every node $u$ in a bottom-up fashion, assigning $g(u)=1$ if and only if $u$ is $(y,d)$-unique or if it is too close to the beginning of the string $S^{\rev}$. If $g(u)=1$, in addition it outputs all strings represented along the incoming edge of $u$, except for substrings which contain the $\$$-sign, i.e.\ suffixes of $S^{\rev}\$$. For details,  see Algorithm~\ref{algo:algo1}.

\medskip

\IncMargin{1em}
\begin{algorithm}[t]
  \SetKwInOut{Input}{input}\SetKwInOut{Output}{output}
  \SetKwProg{Procedure}{procedure}{:}{end}
	\DontPrintSemicolon
  \Input{A colored string $S$, the suffix tree ${\cal T}$ of $S^{\rev}$, and $y\in \Gamma$.}
  \Output{All pairs $(T,d)$ such that $T$ is a $(y,d)$-unique substring of $S$.}
	\LinesNumbered
	\BlankLine
	\For{$d \gets 0$  to $n$}{
	{\sc Unique}($\textit{root},y,d$)}
	\BlankLine
	\Procedure{\sc Unique($u,y,d$)}{
    \If(\tcp*[f]{$u$ is a leaf}){$u$ is a leaf}{
    $i \gets$ $\sn(u)$\;
    \If{$i \leq d$ {\bf or} $f^\rev(i-d)=y$}{$g(u) \gets 1$}
    \Else{$g(u) \gets 0$}
    }
    \Else(\tcp*[f]{$u$ is an inner node}){$g(u) \gets \wedge_{v \text{ child of $u$}} \textsc{Unique($v,y,d$)}$}
    \If{$g(u)=1$}{\label{line:begin changes}
    \If(\tcp*[f]{do not output $\$$-substrings}){$u$ is a leaf}{output $L(u,t)^{\rev}$ for every $t = \stringdepth(parent(u))+1, \ldots, \stringdepth(u)-1$}
    \Else{
    output $(L(u,t)^{\rev},d)$ for every $t = \stringdepth(parent(u))+1, \ldots, \stringdepth(u)$} \label{line:end changes}
    }
    \Return $g(u)$
  }
	\caption{{\sc Algo1}}\label{algo:algo1}
\end{algorithm}
\DecMargin{1em}

{\em Analysis:} For fixed $d$, computing $g$ takes amortized $\Oh(n)$ time over the whole tree, since computing $g(u)$ is linear in the number of children of $u$, and therefore, charging the check whether for a child $v$, $g(v)=1$, to the child node, we get constant time per node. So, for fixed $d$, we have $\Oh(n+K) = \Oh(n^2)$ time, where $K$ is the number of $(y,d)$-unique substrings. Altogether, for $d=0,\ldots,n$, the algorithm takes $\Oh(n^3)$ time.

\begin{example}\label{ex:ex2}
  In the running example (Fig.~\ref{fig:ex1}), for color $y$ and delay $d=3$, the leaf nodes $9,2,7,1,$ and $3$ are marked with $1$, and therefore the only inner node $u$ which gets $g(u)=1$ is the parent of leaves number $9,2,7$. {\sc Algo\ref{algo:algo1}} outputs the $(y,3)$-unique substrings {\tt baca}, {\tt cbaca}, {\tt acbaca}, {\tt cacbaca}, {\tt acacbaca}, {\tt cacacbaca}, {\tt acacacbaca}, {\tt caca}, {\tt acaca}, {\tt ca}, {\tt aca}, {\tt ab}, {\tt cab}, {\tt acab}, {\tt bacab}, {\tt cbacab}, {\tt acbacab}, {\tt cacbacab}, {\tt acacbacab}, {\tt cacacbacab}, {\tt bac}, {\tt cbac}, {\tt acbac}, {\tt cacbac}, {\tt acacbac}, {\tt cacacbac}, {\tt acacacbac.}
\end{example}

{\it Remark:} Note that some of these substrings do not occur even once in a position such that the last character is followed by a $y$ with delay $d=3$. For instance, the only occurrence of the substring {\tt bac} in $S$ is at position $7$, so we would expect to see color $y$ at position $9+3=12$, but the string $S$ ends at position $11$. We will treat this and similar questions in Section~\ref{sec:output}.

\subsection{Outputting only minimally $(y,d)$-unique substrings}\label{sec:minimal}

We next modify Algorithm {\sc Algo1} to output only minimally $(y,d)$-unique substrings. As already noted, the work done by {\sc Algo1} in each node is constant except for the output step, which is proportional to the length of the edge label leading to $u$.

In terms of the suffix tree ${\cal T}$ of $S^{\rev}$, minimality can be translated into conditions on the parent node and on the suffix link parent node (equivalently: the suffix link) in ${\cal T}$.  We first need another definition:

\begin{definition}[Left-minimal nodes, left-minimal labels]
  Let $u$ be a node of ${\cal T} = {\cal T}(S^{\rev})$, different from the root, and let $v = \parent(u)$. We call $u$ {\em left-minimal} for $(y,d)$ if $u$ is $(y,d)$-unique but $v$ is not and the label of the edge $(v,u)$ is not equal to $\$$.
  If $u$ is $(y,d)$-unique and left-minimal, then we can define $\leftmin(u) = x_1\cdot L(v)^{\rev}$, the left-minimal $(y,d)$-unique substring of $S$ associated to $u$, where $x=x_1\cdots x_k\in \Sigma^+$ is the label of edge $(v,u)$.
\end{definition}

\begin{example}\label{ex:ex3}
  In our running example, let node $u$ be the parent of leaf nodes $9,2,7$, i.e.\ $u= \locus({\cal T}, {\tt aca})$. Then $u$ is left-minimal, since it is $(y,3)$-unique but its parent is not. Its left-minimal label is $\leftmin(u) = {\tt ca}$. See Fig.~\ref{fig:ex1}.
\end{example}

It is easy to modify Algorithm~\ref{algo:algo1} to output only left-minimal substrings: Whenever for an inner node $u$ we get $g(u) =0$, then for every child $v$ of $u$ with $g(v) = 1$, we output $\leftmin(v)$ (if defined). This can be done by  replacing lines~\ref{line:begin changes} to~\ref{line:end changes} in Algorithm~\ref{algo:algo1} by:

\medskip
\begin{quote}
\RestyleAlgo{plain}
\IncMargin{1em}
\begin{algorithm}[H]
	\DontPrintSemicolon
	\LinesNumbered
	\BlankLine
  \setcounter{AlgoLine}{11}
    \If{$g(u) = 0$}{
      \For{ each child $v$ of $u$ with $g(v)=1$}{
        \If{$\leftmin(v)$ is defined}{
          output $(\leftmin(v),d)$
        }
      }
    }
\end{algorithm}
\DecMargin{1em}
\RestyleAlgo{ruled}
 \end{quote}

\medskip

\begin{example}\label{ex:ex4}
The resulting algorithm now outputs, for color $y$ and $d=3$, the left-minimal substrings {\tt ca,ab,bac}.
\end{example}

However, we are interested in substrings which are both left- {\em and} right-minimal. While left-minimality can be identified by checking the parent of a node $u$, for right-minimality, Observation~\ref{obs:1} part (3) tells us that we need to check whether the string without the last character is $(y,d+1)$-unique. In ${\cal T}$, this translates to checking the suffix link of the locus of the left-minimal substring $\leftmin(u)$.

 \begin{proposition}\label{prop:minimality}
   Let $u$ be an inner node of ${\cal T} = {\cal T}(S^{\rev})$, different from the root, such that $L(u)^{\rev}$ is $(y,d)$-unique in $S$. Let $v=\textit{parent(u)}$, and $x_1$ be the first character on the edge $(v,u)$. Further, let $t =  \stringdepth(v)+1$, and $(u',t') = \slink(u,t)$. Then the substring $U = x_1\cdot L(v)^{\rev}$ is minimally $(y,d)$-unique in $S$ if and only if $v$ is not $(y,d)$-unique and $u'$ is not $(y,d+1)$-unique.
 \end{proposition}

 \begin{proof}
   For sufficiency, let $U$ be minimally $(y,d)$-unique in $S$. Since $x_1 L(v)^{\rev} = U$, and $U$ is left-minimal, therefore $v$ is not $(y,d)$-unique. Similarly, if $U' = L(u',t')^\rev$, then we have that $U = U'a$, and $u'$ is not $(y,d+1)$-unique by right-minimality of $U$.

   Conversely, since $u$ is $(y,d)$-unique and $v$ is not, by definiton of left-minimality, $U=\leftmin(u)$ is left-minimal $(y,d)$-unique in $S$. Let $U' = L(u',t')^{\rev}$, thus $ U = U'a$, for some character $a\in \Sigma^+$. Since $U'$ is not $(y,d+1)$-unique, therefore $U$ is right-minimal.
\end{proof}

We can use Proposition~\ref{prop:minimality} as follows. Once a left-minimal $(y,d)$-unique node $u$ has been found, check whether $u'$ is $(y,d+1)$-unique, where $u'$ is the node below the locus $\slink(u,\stringdepth(parent(u))+1)$. It is easy to find node $u'$ by noting that $u' = child(\slink(parent(u)),x_1)$, where $x_1$ is the first character of the edge label leading to $u$. But how do we know whether $u'$ is $(y,d+1)$-unique?

The answer is that we will process the distances $d$ in descending order, from $d=n$ down to $d=0$. At the end of the iteration for $d$, we retain the information, keeping a flag on every node $u$ which was identified as $(y,d)$-unique (i.e.\ which had $g(u)=1$). During the iteration for $d-1$, we can then query node $u'$ to find out whether it is $(y,d)$-unique. For details, see Algorithm~\ref{algo:algo2}.

\IncMargin{1em}
\begin{algorithm}[tbp]
  \SetKwInOut{Input}{input}\SetKwInOut{Output}{output}
  \SetKwProg{Procedure}{procedure}{:}{end}
	\DontPrintSemicolon
  \Input{a colored string $S$, the suffix tree ${\cal T}$ of $S^{\rev}$ with suffix links, and $y\in \Gamma$.}
  \Output{all pairs $(T,d)$ such that $T$ is a minimally $(y,d)$-unique substring of $S$.}
	\LinesNumbered
	\BlankLine
	\For{$d \gets n$ downto $0$}{
	{\sc MinUnique}($\textit{root},y,d$)}
	\BlankLine
	\Procedure{\sc MinUnique($u,y,d$)}{
    \If(\tcp*[f]{$u$ is a leaf}){$u$ is a leaf}{
      $i \gets$ $\sn(u)$\;
      \If{$i \leq d$ {\bf or} $f^\rev(i-d)=y$}{$g(u) \gets 1$}
      \Else{$g(u) \gets 0$}
    }
    \Else(\tcp*[f]{$u$ is an inner node}){$g(u) \gets \wedge_{v       \text{ child of $u$}} \textsc{MinUnique($v,y,d$)}$}\label{line:minimality  base}
    \If(\tcp*[f]{outputting minimal substrings for children}){$g(u)=0$}{
      \For{each child $v$ with $g(v)=1$}{
      \If{$\leftmin(v)$ is defined}{
        $(v',t) \gets \slink(v,\stringdepth(u)+1)$ \;
        \If(\tcp*[f]{flag from previous round}){$v'$ is not $(y,d+1)$-unique}{
          output $(\leftmin(v),d)$}\label{line:output base}
        }
      }
    }
    \Return $g(u)$

  }
	\caption{{\sc Algo2}}\label{algo:algo2}
\end{algorithm}
\DecMargin{1em}

\medskip

\begin{example}\label{ex:ex5}
  In the running example, we know from the previous round that the only nodes that are $(y,4)$-unique are the leaves number $4,2,1,10,3,$ and $8$. We can now deduce that the substring ${\tt ca}$ is right-minimal, because  $u = \locus({\tt ca})$ is not $(y,4)$-unique, and $\slink(\locus({\cal T},{\tt ca}^{\rev})) = (u,1)$. Looking at the string $S$ we see that ${\tt ca}$ is indeed right-minimal, since {\tt c} is not $(y,3)$-unique: it has an occurrence, in position 6, which is not followed by a $y$ but by an $x$ at position $10=6+4$ (delay $4$). Similarly, the left-minimal substring {\tt ab} is also right-minimal, since its suffix link is not $(y,4)$-unique, while the left-minimal substring {\tt bac} is not $(y,3)$-unique, because its suffix link is $(y,4)$-unique, see Fig.~\ref{fig:ex1}. %
\end{example}

\medskip

{\em Analysis: } For fixed $d$, the time spent on each leaf is constant (lines 5 to 10 in {\sc Algo2}); we charge the check of $g(v)$ in line 12 to the child $v$, as well the work in lines 14 to 18 (computing $\leftmin(v)$ and checking the flag on $v'$ from the previous round); these are all constant time operations, so we have amortized constant time per node, and thus
$\Oh(n)$ time for fixed $d$. Therefore, the total time taken by Algorithm~\ref{algo:algo2} is $\Oh(n^2)$.

\subsection{An algorithm for all colors}\label{sec:general_colors}

In some situations, one is interested in all minimally $(y,d)$-unique substrings, {\em for any color $y$}. Our third algorithm deals with this case (Problem 2). It is similar to {\sc Algo2}, except it uses a different coloring function $g'$. The new function  $g':V \to \Gamma \cup \{*,0\}$, is defined as follows:

\medskip

\begin{itemize}
	\item   for a leaf $u$ with leaf number $\sn(u)=i$:
	\[ g'(u) = \begin{cases}
	f^{\rev}(i-d)  & \text{ if  $i-d > 0$},\\
	* & \text{ if $i-d \leq 0$}
	\end{cases}\]
	\item for an inner node $u$:
	\[ g'(u) =
	\begin{cases}
	* & \text{ if for all children $v$ of $u$: $g'(v) = *$,}\\
	y \in \Gamma & \text{ if $u$ has at least one child $v$ with $g'(v)=y$ and }\\
	& \quad \text{ for all other children $v'$ of $u$: $g'(v') \in \{y,*\}$ },\\
	0 & \text{ otherwise.}
	\end{cases} \]
\end{itemize}

\medskip

Thus a node $u$ is colored $y$ if and only if all leaves of the subtree rooted in $u$ are either colored $y$ or $*$, and at least one leaf is colored $y$. We refer to such a node as {\em monochromatic}. A node is colored $*$ if all leaves in the subtree are within $d$ of the beginning of $S^{\rev}$; such a node can have monochromatic ancestors in the tree. Finally, a node is colored $0$ if in its subtree there are at least two leaves which are colored by different colors from $\Gamma$. For a node colored $0$, all of its ancestors are also colored $0$.

\begin{example}\label{ex:ex6}
  In our example,  for $d=3$, the leaves 11, 4, and 8 are colored $z$, the leaves 9 and 7 are colored $y$, the leaves 5, 10, and 6 are colored $x$, and the leaves 1,2, and 3 are colored $*$. The only monochromatic inner nodes are  $\locus({\tt b})$ (colored $x$), and $\locus({\tt aca})$ (colored $y$), while all others are colored $0$. See Figure~\ref{fig:allcolors}.

\begin{figure}
    \centering
    \begin{subfigure}[b]{\textwidth}
        \includegraphics[width=\textwidth]{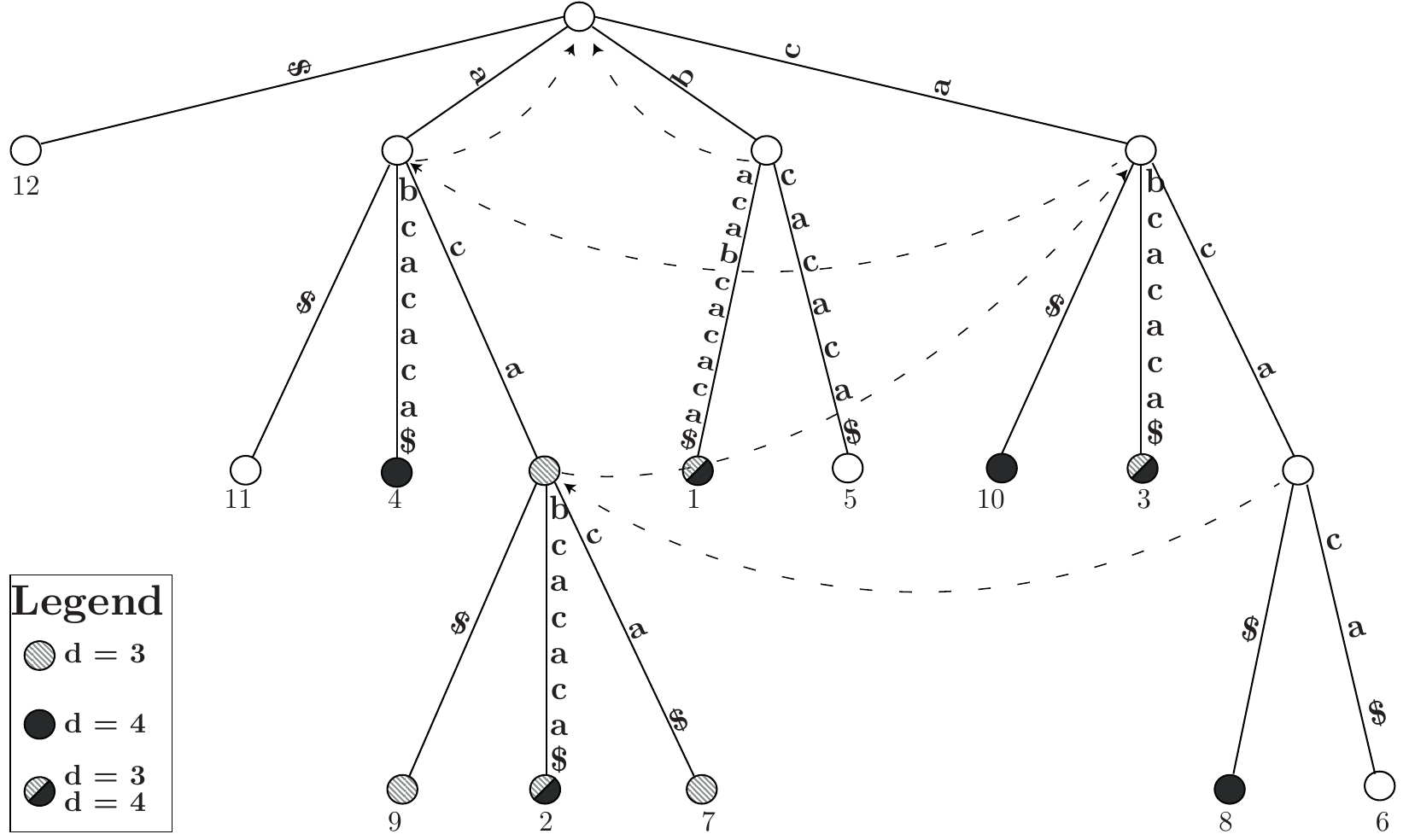}
        \caption{%
        The nodes are colored according to function $g$ for the character $y$, for $d = 3$ (dashed) and for $d = 4$ (solid), see Example~\ref{ex:ex2}.\label{fig:ycolor}}
    \end{subfigure}

    \begin{subfigure}[b]{\textwidth}
      \includegraphics[width=\textwidth]{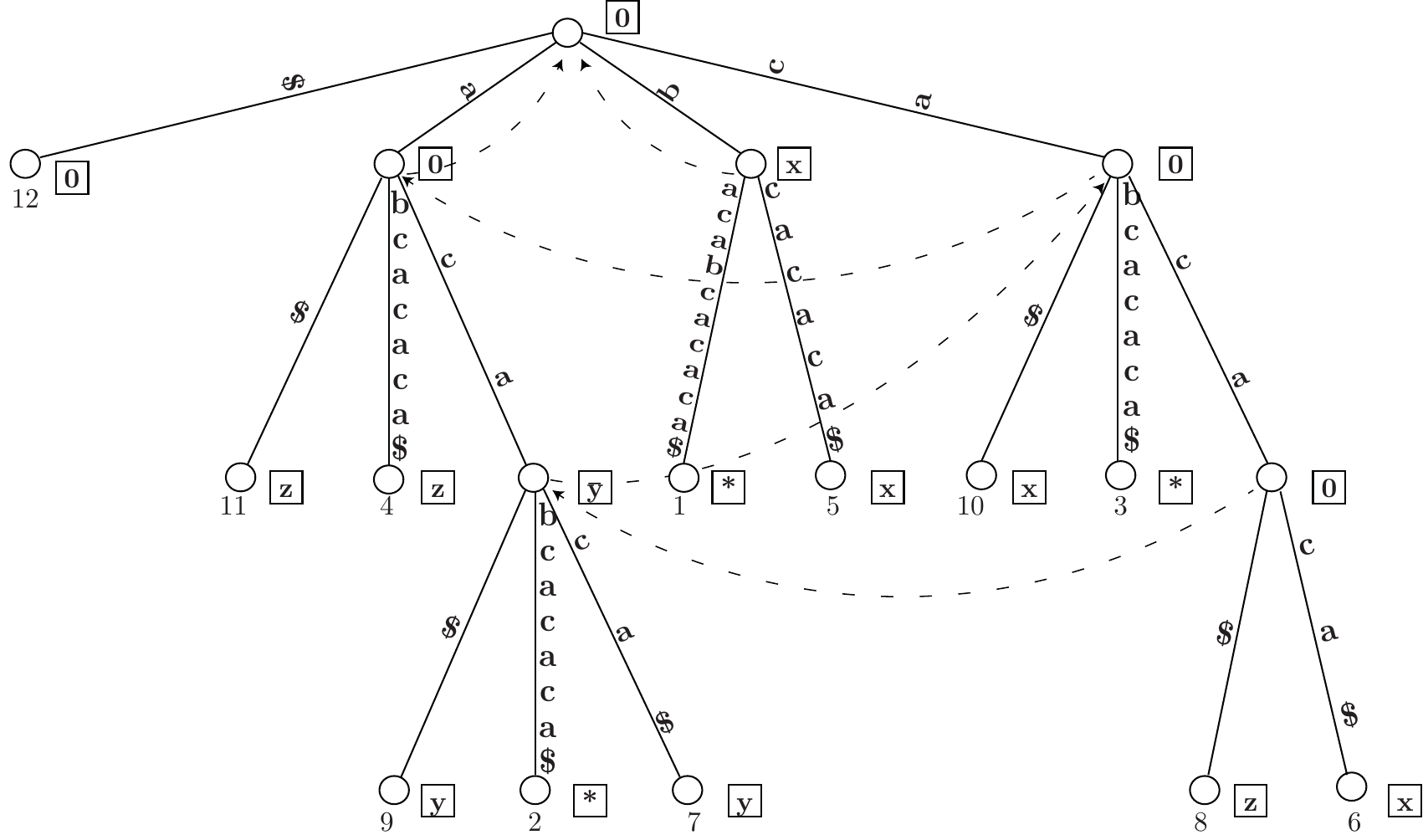}
      \caption{%
      The nodes are marked according to function $g'$ for $d = 3$, see Example~\ref{ex:ex6}.\label{fig:allcolors}}
    \end{subfigure}
    \caption{The suffix tree ${\cal T}$ of the reverse string $S^{\rev} = {\tt bacabcacaca}$, where $S = {\tt acacacbacab}$, see Example~\ref{ex:ex1}. For clarity, the edges carry the label itself rather than a pair of pointers into the string. Suffix links are drawn as dotted directed edges.\label{fig:ex1}}
\end{figure}

\end{example}

\IncMargin{1em}
\begin{algorithm}[t!]
  \SetKwProg{Procedure}{procedure}{:}{end}
  \SetKwInOut{Input}{input}\SetKwInOut{Output}{output}
	\DontPrintSemicolon
  \LinesNumbered
  \Input{a colored string $S$, and the suffix tree ${\cal T}$ of $S^{\rev}$ with suffix links.}
  \Output{all triples $(T,y,d)$ such that $T$ is a minimally $(y,d)$-unique substring of $S$}
	\BlankLine
    \For{$d \gets n$ downto $0$}{
      {\sc AllColorsMinUnique}($\textit{root},d$)
    }
	\BlankLine
	\Procedure{\sc AllColorsMinUnique($u,d$)}{
    \If(\tcp*[f]{$u$ is a leaf}){$u$ is a leaf}{
    $i \gets$ $\sn(u)$\;
    \If{$i \leq d$}{$g'(u) \gets *$}
    \Else{$g'(u) \gets f^{\rev}(i-d)$}
    }
    \Else(\tcp*[f]{$u$ is an inner node}){
    $X \gets \{$ {\sc AllColorsMinUnique}($v,y,d$) $\mid v \text{ child of } u\}$\;
    \If{$X = \{*\}$}{$g'(u) \gets *$}
    \Else{\If{$X = \{y\}$ {\bf or} $X = \{y,*\}$ with $y\in \Gamma$}
    {$g'(u) \gets y$}
    \Else{$g'(u) \gets 0$}
    }}
    \If(\tcp*[f]{outputting minimal substrings for children}){$g'(u)=0$}{\label{line:minimality  base-all}
      \For{each child $v$ with $g'(v) = y \in \Gamma$}{
        \If{$\leftmin(v)$ is defined}{
          $(v',t) \gets \slink(v,\stringdepth(u)+1)$ \;
          \If(\tcp*[f]{flag from previous round}){$v'$ is not $(y,d+1)$-unique}{
            output $(\leftmin(v),y,d)$
          }
        }
      }
    }
    \Return $g'(u)$

  }
	\caption{{\sc Algo3}%
  }\label{algo:algo3}
\end{algorithm}
\DecMargin{1em}

\medskip

Thus, Algorithm~\ref{algo:algo3} finds all minimally $(y,d)$-unique substrings for all colors simultaneously, using the same ideas as {\sc Algo2}. The main difference is that now the coloring function $g'$ is not binary, and accordingly, the information we have to store from the previous round (which will be needed to decide whether the substring is right-minimal) is no longer binary. See {\sc Algo3} for details.

{\em Analysis:} The algorithm has $n$ iterations, every iteration takes $\Oh(n)$ time, so altogether we have again $\Oh(n^2)$ time.

\medskip

Therefore, if the color of interest is not part of the input, we can solve the problem in $\Oh(n^2)$ time, which is also a worst-case lower bound on the output size, see Sec.~\ref{sec:basics}. However, if the color $y$ is part of the input, then this algorithm can be further improved. We will present this improvement in the next section.

\section{Skipping Algorithm}\label{sec:skipping}

  In this section, we discuss the discovering of $(y,d)$-unique substrings that are minimal. As in the baseline algorithm, we build the suffix tree ${\cal T}(S^\rev)$ and, intuitively, we navigate it discovering all left-minimal $(y,d)$-unique substrings one by one, reporting only those that are minimal. Thus, according to Proposition~\ref{prop:minimality}, we have to discover all left-minimal $(y,d+1)$-unique substrings before discovering left-minimal $(y,d)$-unique substrings.

  To this end, fixing $\ell$, for each node $u$ of ${\cal T}(S^\rev)$, we determine the largest delay $d$ smaller than $\ell$ such that $L(u)^\rev$ can be $(y,d)$-unique, denoted by $h(u,\ell)$.
  We consider four different cases:
  \begin{itemize}
	 \item If $u$ is a leaf, then $L(u)^\rev$ is the $j$-prefix of $S$, where $j = n-\sn(u)+1 = |L(u)|$
	\begin{itemize}
		\item If $\sn(u) < \ell$, then $j+\ell - 1 > n$ thus $L(u)^\rev$ is $(y,\ell-1)$-unique since the position of the color is beyond the end of the string, thus $h(u,\ell) = \ell-1$.
		\item If $\sn(u) \geq \ell$ and there exists an $i < \ell$ such that $f(j+i) = y$, then the highest possible value $d<\ell$ such that $L(u)^\rev$ is $(y,d)$-unique is given by the position of the furthest occurrence of $y$ within a distance of $\ell-1$ from $j$, thus $h(u,\ell) = \max\{ i < \ell \mid f(j+i) = y \}.$
		\item Otherwise, if such $i$ does not exists, we set $h(u,\ell) = -1$.
	\end{itemize}
	\item If $u$ is an internal node of ${\cal T}(S^\rev)$, then let $k = \min\{ h(v,\ell) \mid v \text{ child of } u \}$, since it is not possible that $L(u)^\rev$ is $(y,d')$-unique, for any $k < d' < \ell$, thus $h(u,\ell) = k$.
  \end{itemize}
	When $u$ is an inner node in general, we do not know if $L(u)^\rev$ is $(y,d)$-unique for $d = h(u,\ell)$, unless for all nodes $v$ in the subtree rooted in $u$, there exists an $\ell_v$ such that $h(u,\ell) < \ell_v \leq \ell$ and $h(v,\ell_v) = h(u,\ell)$. This is true if $h(v,d+1) = h(u,\ell)$ for all $v$.

  The definition of $h(u,\ell)$ is as follows: $$ h(u,\ell) = \begin{cases} \ell -1 & \text{ if } u \text{ is a leaf and } \sn(u) < \ell,\\ \max\{ i < \ell \mid f(n-\sn(u)+1+i) = y \} & \text{ if } u \text{ is a leaf and such } i \text{ exists, } \\ \min\{ h(v,\ell) \mid v \text{ child of } u \} & \text{ if } u \text{ is an inner node,} \\ -1 & \text{ otherwise. } \end{cases}$$

 	\begin{lemma}\label{lemma:nodes}
 		Let $u$ be a node of ${\cal T}(S^\rev)$, fix $d$, $h(u,d+1) = d$ if and only if $u$ is $(y,d)$-unique.
 	\end{lemma}

 	\begin{proof}
 		We first prove that if $h(u,d+1) = d$ then $u$ is $(y,d)$-unique. We consider two cases.
 		If $u$ is a leaf, then, by definition of $h(u,d+1)$, we have that $u$ is $(y,d)$-unique.
 		If $u$ is an inner node, then $d = \min\{h(v,d+1) \mid v $ child of $u \}$. Since for all nodes $v$, $h(v,d+1) \leq d$, then for all children $v$ of $u$, we have that $h(v,d+1) = d$. In particular, this holds for all leaves in the subtree rooted in $u$, thus $u$ is $(y,d)$-unique.
 		
 		We first prove that if $u$ is $(y,d)$-unique then $h(u,d+1) = d$. We consider again two cases.
 		If $u$ is a leaf, then, by definition of $(y,d)$-unique, we have that either $\sn(u) < d+1$ or $f(n-\sn(u)+1+d) = y$. Thus, in both cases, $h(u,d+1) = d$.
 		If $u$ is an inner node, then, by Lemma~\ref{lemma:1}, all leaves in the subtree rooted in $u$ are $(y,d)$-unique. Thus, for the previous case, for all leaves $v$ in the subtree rooted in $u$ we have that $h(v,d+1) = d$, thus $h(u,d+1) = d$.
 	\end{proof}

  To evaluate $h(u,\ell) = \max\{ i < \ell \mid f(j+i) = y \}$ when $u$ is a leaf and such $i$ exists, we define a bitvector $b_y[1,2n]$ such that $b_y[i] =1$ only if $f(i) = y$ or $i > n$.
  We preprocess $b_y$ for $\Oh(1)$-time {\tt rank} and {\tt select} queries~\cite{clark1997compact}.
  Given node $u$ with $\sn(u) \geq \ell$, let $ j = n - \sn(u) + 1$. 
  We have that $h(u,\ell) = \max\{{\tt select}(b_y, {\tt rank}(b_y, j+\ell)) -j, -1\}$.

  \begin{example}\label{ex:h function}
		In our running example, whose suffix tree is depicted in Figure~\ref{fig:ex1}, let us consider the node $u$ corresponding to the substring {\tt aca} in the string $S^\rev$. In order to compute $h(u,9)$, we have to recursively compute the $h$ function for all children of $u$. Let $v$, $s$, and $t$ be the leaves corresponding to the $9$-th, $2$-nd, and $7$-th suffix of $S^\rev$, respectively. If we remove the dollar character from the end of the string $S^\rev$, then the $9$-th, $2$-nd, and $7$-th suffix of $S^\rev$ corresponds to the $3$-rd, $10$-th, and $8$-th prefix of $S$, respectively.
		We have that $h(s,9) = h(t,9) = 8$, since the furthest possible $y$ at distance smaller than $9$ from the $10$-th and $8$-th prefix of $S$ are beyond the end of the string. While, the furthest possible $y$ at distance smaller than $9$ from the $3$-rd prefix of $S$ is at distance $5$. Thus $h(v,9) = h(u,9) = 5$. The intuition is that the highest possible $d$, smaller than $9$ such that the substring {\tt aca} can be $(y,d)$-unique cannot be larger than $5$, since there is an occurrence of {\tt aca} that has no $y$'s at distance between $6$ and $8$.

		Let us now compute $h(u,3)$. We have that $h(s,3) = 2$, since the furthest possible $y$ at distance smaller than $9$ corresponding to the $10$-th prefix of $S$ is beyond the end of the string. For the $8$-th prefix of $S$ we have that the furthest $y$ at distance smaller than $3$ is at distance $1$, thus $h(t,3) = 1$. While, for the leaf $v$ there is no $y$ at distance smaller than $3$, thus $h(v,3) = -1$. Hence, we have that $h(u,3) = -1$.
	\end{example}

We use the $h(u,\ell)$ function in the following way, during the discovery process of all $(y,d)$-unique substrings of $S$, provided that we have already discovered all $(y,d+1)$-unique substrings of $S$. Let $\ell = d+1$ , for all nodes $u$ of ${\cal T}(S^\rev)$ we store the values $h(u,\ell)$. We discover the minimally $(y,d)$-unique substrings of $S$, finding all nodes $u$ such that $h(u,\ell) = d$. Among those, the nodes that are also left-minimal are those nodes $u$ such that, $h(\parent(u), \ell) < d$. We then check if $u$ is also right-minimal by checking if its suffix-link parent is $(y,d+1)$-unique, as in Algorithm~\ref{algo:algo2}.

The key idea of the skipping algorithm is to keep the values $h(u,\ell)$ updated during the process. Let $H(u)$ be the array that, at the beginning of the discovery of all $(y,d)$-unique substrings of $S$, stores the values $h(u,\ell)$. We want to keep the array $H$ updated in a way such that, after we discovered all $(y,d)$-unique substrings of $S$, for all nodes $u$, $H(u) = h(u,\ell-1)$.
Thus, once we discover that a node $u$ is left-minimal $(y,d)$-unique, we update the value of $H(u)= h(u,\ell-1)$.
We then update the following values:
\begin{itemize}
	\item for all nodes $v$ in the subtree rooted in $u$, we update the values $H(u) = h(u,\ell-1)$.
	\item for all nodes $p$ ancestors of $u$, we update the values $H(p) = \min(H(p), h(u,\ell-1))$
\end{itemize}

\begin{lemma}\label{lemma:nodes2}
    Given ${\cal T}(S^\rev)$, fix $d$, for all nodes $u$ of  ${\cal T}(S^\rev)$, let $H(u) = h(u,d+1)$. If for all nodes $u$ such that $H(u) = d$ we (i) set $H(u) = h(u,d)$, and (ii) for all ancestors $p$ of $u$, set $H(p) = \min\{H(p), h(u,d)\}$, then, for all nodes $u$ of ${\cal T}(S^\rev)$, $H(u) = h(u,d)$.
\end{lemma}

\begin{proof}
    Let $H'(u)$ be the array after all the updates. We now proceed by cases.
    If $u$ is a leaf and $H(u) = d$, then we set $H(u) = h(u,d)$. Since $u$, is a leaf, the value $H(u)$ is not modified by any other operation, thus $H'(u) = h(u,d)$. If $u$ is a leaf and $H(u) < d$, then, by definition of $h(u,d+1)$, we have that $h(u,d) = h(u,d+1)$. Since $u$ is a leaf and its value $H(u)$ is not modified by ant other operation, $H'(u) = h(u,d)$.
	 		
    If $u$ is an internal node and $H(u) = d$, then we set $H(u) = h(u,d)$. Moreover, for all nodes $v$ in the subtree rooted in $u$, we have that $H(v) = d$ and when they perform (ii), they update $H(u) = \min\{H(u),h(v,d)\} = \min\{h(u,d),h(v,d)\} = h(u,d)$ by definition of $h(u,d)$. thus, $H'(u) = h(u,d)$. Finally, if $u$ is an internal node and $H(u) <d $, then if for all nodes $v$ in the subtree rooted in $u$, $H(v) < d$, we have that $h(u,d+1) = h(u,d)$. Otherwise, let ${\cal L}$ be the set of nodes $v$ in the subtree rooted in $u$ such that $H(v) = d$. Each node $v \in {\cal L}$ change the value of $H(u)$ as $\min\{H(u), h(v,d)\}$, thus we have that $H'(u) = \min\{h(v,d) \mid v $ leaf in the subtree rooted in $u\} = h(u,d)$.
\end{proof}

To efficiently find all nodes $u$ such that $h(u,\ell) = d$ and $h(\parent(u), \ell) < d$, we use a {\em  maximum-oriented indexed priority queue}, storing the values of $H(u)$ as keys and $iBFS(u)$ as index. Under this condition, if two nodes have the same key value, then parents have higher priority than their children in $IPQ$.
We keep the priority queue updated using a {\tt demote} operation while we discover left-minimal nodes and we update the values of the array $H$ stored as keys of $IPQ$.
Algorithm~\ref{algo:lazy} shows how to compute $h(u,\ell)$ for a given node $u$, and how we update the values of the keys in the $IPQ$ for all children $v$ of $u$.

  \IncMargin{1em}
  \begin{algorithm}[tbp]
    \SetKwInOut{Input}{input}\SetKwInOut{Output}{output}
    \SetKwProg{Procedure}{procedure}{:}{end}
  	\DontPrintSemicolon
    \Input{A node $u$ in the suffix tree ${\cal T}(S^\rev)$ and integer $\ell$.}
    \Output{Maximum delay $d<\ell$ such that $L(u)^\rev$ can be $(y,d)$-unique.}
  	\LinesNumbered
  	\BlankLine
    \Procedure{$h(u,\ell)$}{
      $ min_d \gets \ell -1$\;
      \If{$u$ is a leaf}{
      $ j \gets n - \sn(u) + 1$\;
      $min_d \gets \max\{{\tt select}(b_y, {\tt rank}(b_y, j+\ell)) -j, -1\}	$\;
      }
      \Else{
      \ForAll{children $v$ of $u$}{
      $ d = h(v,\ell)$\;
      \If{$min_d < d$}{
      $ min_d \gets d$\;
      }
      }
      }
      $ IPQ.${\tt demote($u,min_d$)}\;
      \Return $min_d$\;
    }

  	\caption{{\sc Highest possible value of $d$. }}\label{algo:lazy}
  \end{algorithm}
  \DecMargin{1em}

  The skipping algorithm summarized in Algorithm~\ref{algo:mine} initializes priority queue $IPQ$ by inserting all nodes of ${\cal T}(S^\rev)$ with key $n+1$. We then repeat the following as long as there exists a node with non-negative key: extract the max element $(u,\ell)$ of $IPQ$; decide whether or not it has to be reported, i.e. if it is right-minimal; apply Algorithm~\ref{algo:lazy} to update the key values of all nodes in the subtree of $u$ and then update the values of the keys of all ancestors of $u$.

  \IncMargin{1em}
  \begin{algorithm}[tbp]
    \SetKwInOut{Input}{input}\SetKwInOut{Output}{output}
    \SetKwProg{Procedure}{procedure}{:}{end}
  	\DontPrintSemicolon
    \Input{A colored string $S$, and a color $y \in \Gamma$}
    \Output{All minimal $(y,d)$-unique substrings of $S$.}
  	\LinesNumbered
  	\BlankLine

    \ForAll{ nodes $v$ of ${\cal T}(S^\rev)$ }{
      $ IPQ$.{\tt insert($v,n+1$)} \;
    }
    \While{ $IPQ.${\tt allNegative()}$ = false$}{
      $(u,d) \gets IPQ.${\tt max()}\;
        $ (u', t) = slink(u, sd(parent(u))+1) $\;
         \If(\tcp*[f]{flag from previous round}){$u'$ is not $(y,d+1)$-unique}{\label{line:right-minimality check begin}
        	output $(d,\leftmin(u))$\label{line:output skip}\;
        }
      $ min_d = h(u, d)$\label{line:h'}\;
      \ForAll{ancestors $v$ of $u$}{
        \If{$ IPQ.${\tt keyOf($v$)} $>  min_d$}{
          $ IPQ.${\tt demote($v,min_d$)}\;
        }
      }

    }

  	\caption{{\sc Skipping}}\label{algo:mine}
  \end{algorithm}
  \DecMargin{1em}

{\em Analysis:} For all nodes $u$ in ${\cal T}(S^\rev)$, the key value associated to $u$ in $IPQ$ is initially $n+1$. Each time Algorithm~\ref{algo:mine} and Algorithm~\ref{algo:lazy} visit a node, the key value of $u$ in $IPQ$ is decreased (via {\tt demote()}) until it becomes negative. Thus, for each node we perform at most $n+1$ {\tt demote()} operations. Since the number of nodes in ${\cal T}(S^\rev)$ is linear in $n$, Algorithm~\ref{algo:mine} runs in $\Oh(n^2\log(n))$ time.

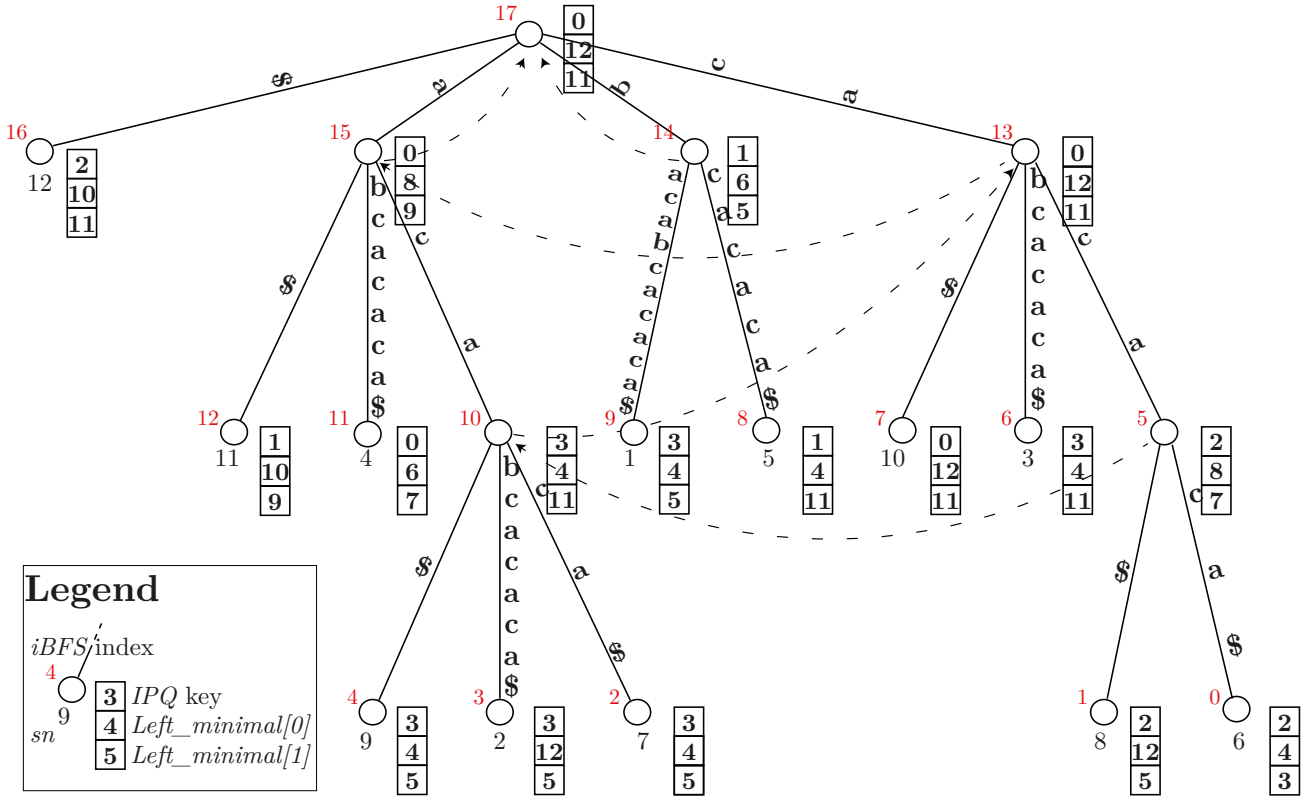
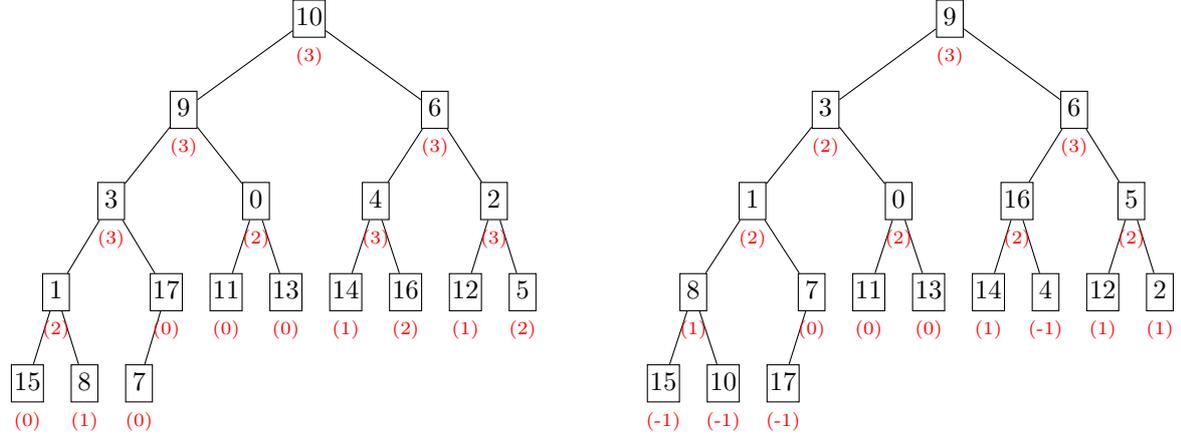
\begin{figure}
	\centering
	    \begin{subfigure}[b]{\textwidth}
	    	\includegraphics[width=\textwidth]{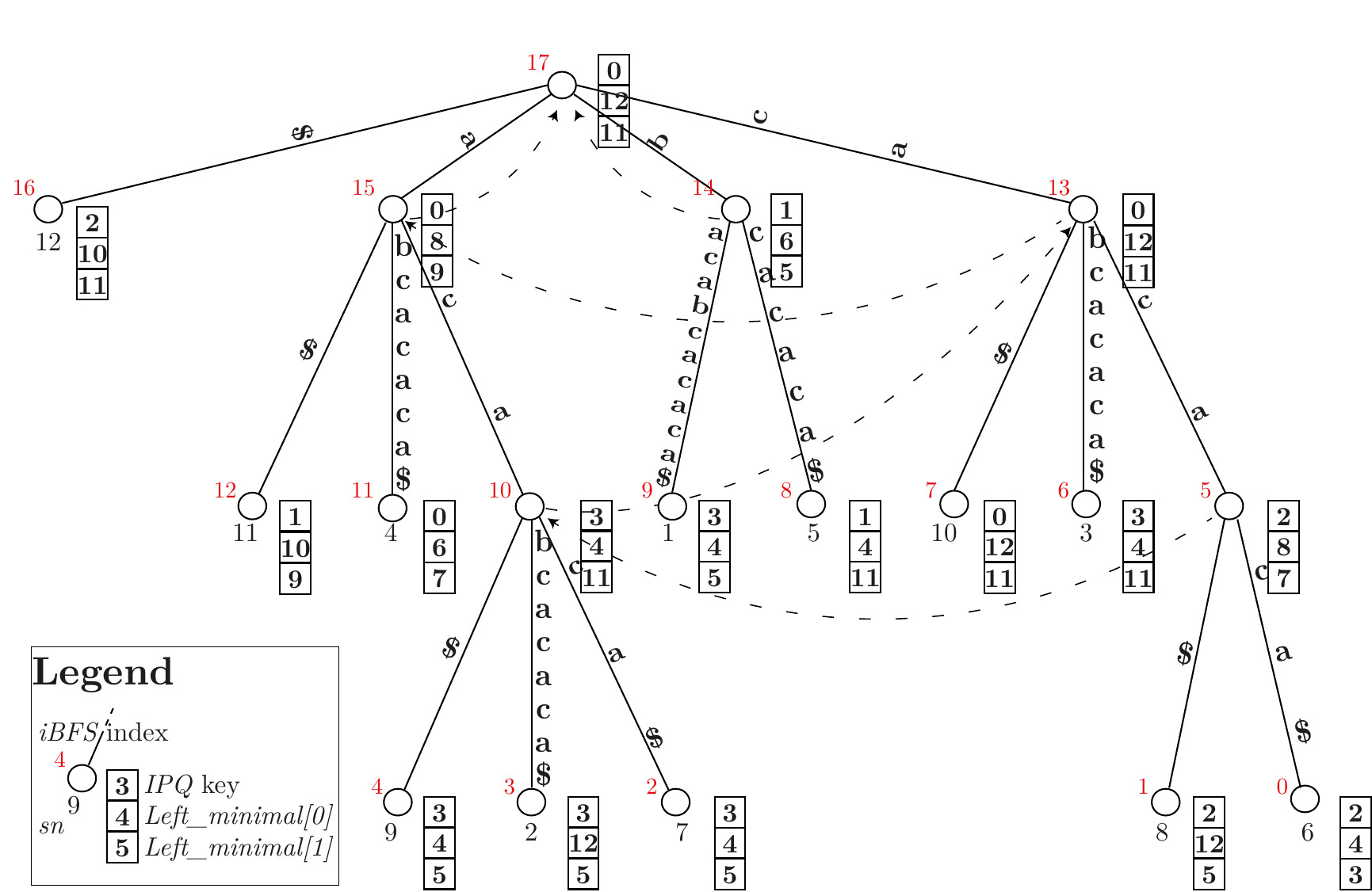}
	    	\caption{The suffix tree of ${\cal T}$ of $S^{\rev}$, reporting the values of $IPQ$, $left\_minimal[0]$, and $left\_minimal[1]$ for each node, after $48$ iterations of Algorithm~\ref{algo:lazy}.\label{fig:h-function}}
	    \end{subfigure}
	\begin{subfigure}[b]{0.45\textwidth}
		\begin{forest}
			for tree={
				rectangle,
				minimum width=1em,
				minimum height=1em,
				inner sep=1pt,
				l sep=2em,
				s sep=1em,
				draw,
				align=center,
			},
			key val/.style={
				tikz+={
					\node [anchor=mid west, red,  below=0pt of .south, font=\scriptsize]  {(#1)};
				},
			}
			[$10$,key val=3
				[$9$,key val=3
					[$3$, key val=3
						[$1$, key val=2
							[$15$, key val=0]
							[$8$, key val=1]
						]
						[$17$, key val=0
							[$7$, key val=0]
							[,phantom,]
						]
					]
					[$0$, key val=2
						[$11$, key val=0]
						[$13$, key val=0]
					]
				]
				[$6$, key val=3
					[$4$, key val=3
						[$14$, key val=1]
						[$16$, key val=2]
					]
					[$2$,key val=3
						[$12$, key val=1]
						[$5$, key val=2]
					]
				]
			]
		\end{forest}
		\caption{Indexed priority queue $IPQ$ after $48$ iterations of Algorithm~\ref{algo:lazy}.}\label{fig:impq-1}
	\end{subfigure}
	\hspace{1em}
	\begin{subfigure}[b]{0.45\textwidth}
		\begin{forest}
			for tree={
				rectangle,
				minimum width=1em,
				minimum height=1em,
				inner sep=1pt,
				l sep=2em,
				s sep=1em,
				draw,
				align=center,
			},
			key val/.style={
				tikz+={
					\node [anchor=mid west, red,  below=0pt of .south, font=\scriptsize]  {(#1)};
				},
			}
			[$9$,key val=3
				[$3$,key val=2
					[$1$, key val=2
						[$8$, key val=1
							[$15$, key val=-1]
							[$10$, key val=-1]
						]
						[$7$, key val=0
							[$17$, key val=-1]
							[,phantom,]
						]
					]
					[$0$, key val=2
						[$11$, key val=0]
						[$13$, key val=0]
					]
				]
				[$6$, key val=3
					[$16$, key val=2
						[$14$, key val=1]
						[$4$, key val=-1]
					]
					[$5$,key val=2
						[$12$, key val=1]
						[$2$, key val=1]
					]
				]
			]
		\end{forest}
		\caption{Indexed priority queue $IPQ$ after $49$ iterations of Algorithm~\ref{algo:lazy}.}\label{fig:impq-2}
	\end{subfigure}
	\caption{Top (\ref{fig:h-function}): The suffix tree of ${\cal T}$ of the reverse string $S^{\rev} = {\tt bacabcacaca}$, reporting the values of $IPQ$, $left\_minimal[0]$, and $left\_minimal[1]$ for each node, after $48$ iterations of Algorithm~\ref{algo:lazy}. In red, in the upper left of each node, we report the reverse index BFS of the node, below each leaf we report the associated suffix number, on the right of the node we report the values of $IPQ$, $left\_minimal[0]$, and $left\_minimal[1]$. Bottom: The indexed priority queue $IPQ$ after $48$~(\ref{fig:impq-1}) and $49$~(\ref{fig:impq-2}) iterations of Algorithm~\ref{algo:lazy}. In the nodes of the priority queue we have the index of the nodes of the suffix tree ${\cal T}(S^\rev)$ numbered in the reverse index BFS. Below each node, in red and in brackets, the value of the key associated to each index.}
	\label{fig:example h function}
\end{figure}

\begin{example}
	In our running example, we want to report all minimally $(y,d)$-unique substrings of the colored string, for the character $y$. Using Figure~\ref{fig:example h function}, we now show how we discover that the substring {\tt ca} is $(y,3)$-unique. We show in Figure~\ref{fig:impq-1} the indexed priority queue after $48$ iterations of Algorithm~\ref{algo:lazy}. The maximum element in the indexed priority queue $IPQ$ is the node of ${\cal T}(S^\rev)$ corresponding to the $10$-th node in the reverse index BFS of the tree, as shown in Figure~\ref{fig:h-function}. The associated key value of the maximum element is $3$, which means that the corresponding substring is left-minimal $(y,3)$-unique. In order to decide if the corresponding substring is also right-minimal, we check if the suffix link parent of the node number $10$, which is the node number $13$, is left-minimal for $d = 4$. The value of the last even value such that the node number $13$ has been left-minimal is set to $12$. Thus the node $10$ is minimally $(y,3)-unique$ and has to be reported. We now compute the $h(u,\ell)$ function for the node number $10$, $u$, and $\ell = 3$ . As shown in Example~\ref{ex:h function}, we have that the $h$ function for the nodes number $2$, $3$, and $4$ are $1$, $2$, $-1$. Thus, the value of the $h$ function for the node number $10$ is $-1$. We then update the values of all parents of the node number $10$. This results in an update of the values of the indexed priority queue $IPQ$ as reported in Figure~\ref{fig:impq-2}.
\end{example}

\subsection{Right-minimality check}
According to Proposition~\ref{prop:minimality}, in order to decide if a node is left-minimal $(y,d)$-unique, we have to check that the suffix link parent $u' = \slink(u)$ is not a left-minimal $(y,d+1)$-unique node.
Since we discover $(y,d)$-unique substrings in decreasing order of $d$,  it is enough to store, for each node, the previous value of $d$ such that the node is left-minimal.

Given a node $u$, we store this information in two arrays, indexed by the $iBFS(u)$ values. In one array we store the last even values of $d$ such that the node was left-minimal. In the other array we store the last odd values of $d$ such that the node was left-minimal. This prevents possible overwriting of information, e.g., let $v = \slink(u)$ such that $v$ is left-minimal $(y,d+1)$-unique and left-minimal $(y,d)$-unique node. Let us assume that $v$ is processed before the node $u$ that is left-minimal $(y,d)$-unique. If we had only one array holding the information of the last value of $d$ such that a node was left-minimal, then this value for $v$ would now be $d$, instead of $d+1$. Thus, we would erroneously conclude that $v$ is also right-minimal, hence that it is minimally $(y,d)$-unique. Using one array to store even values of $d$ and one array to store odd values of $d$, we avoid this problem, since $v$ updates the array associated to the parity of $d$, while $u$ queries the one associated to the parity of $d+1$.

We can replace lines~\ref{line:right-minimality check begin} to~\ref{line:output skip} with the following lines of code, where we set at the beginning $left\_minimal[b][u] = \infty$ for all $b = \{0,1\}$ and for all nodes $u$.

\medskip

\begin{quote}
	\RestyleAlgo{plain}
	\IncMargin{1em}
	\begin{algorithm}[H]
		\DontPrintSemicolon
		\LinesNumbered
		\BlankLine
		\setcounter{AlgoLine}{5}
	 $ report \gets (left\_minimal[(d+1)\bmod 2][u'] \neq d+1)$\;
	 \If{$report$}{
		 output $(d,\leftmin(u))$
	 }
	 $ left\_minimal[d\bmod 2][u] \gets d$
	\end{algorithm}
	\DecMargin{1em}
	\RestyleAlgo{ruled}
\end{quote}

\medskip

See Figure~\ref{fig:h-function} for an example of the values of the arrays $left\_minimal[0]$ and $left\_minimal[1]$.

\section{Output restrictions and algorithm improvement}\label{sec:output}

We now discuss some practically-minded output restrictions. They can be implemented as a filter to the output, thus discarding some solutions, but if they are considered as part of the problem specification, then they lead to an improvement for the skipping algorithm.

Note that our definition of $(y,d)$-unique allows that a substring occurs only once, or that none of its occurrences is followed by a $y$ with delay $d$, because they are all close to the end of string. We now restrict our attention to $(y,d)$-unique substrings with at least two occurrences followed by $y$ with delay $d$.

Given a colored string $S$, let $T$ be minimally $(y,d)$-unique. We report $(T,d)$ if and only if the following holds:

\begin{enumerate}
	\item There are at least two occurrences of $T$ in $S$.
	\item Let $i$ be the second smallest occurrence of $T$ in $S$, then $i+|T|-1+d \leq n$.
\end{enumerate}

A substring $T$ that satisfies the above conditions is called a {\em real type} minimally $(y,d)$-unique substring.
Note that any algorithm that computes all minimally $(y,d)$-unique substrings can be easily modified to output only those that are of real-type, by checking the two conditions before outputting (in line~\ref{line:output skip} of Algorithm~\ref{algo:mine}, resp.\ line~\ref{line:output base} of Algorithm~\ref{algo:algo2}): if the node $u$ is not a leaf and the value of the second largest suffix of $S^\rev$ in the subtree rooted in $u$ is greater than or equal to $d$. Since each node $u$ in the suffix tree ${\cal T}(S^\rev)$, corresponds to an interval $[i,j]$ of the suffix array of $S^\rev$, we can find the second largest suffix using a range maximum query $rMq$ data structure~\cite{FH11} built on the suffix array of $S^\rev$ (see Sec~\ref{sec:basics}).

\medskip

We now turn to the {\tt skipping} algorithm specifically, which we can modify such that it only computes real-type solutions. 
The $h(u,\ell)$ function is used in Algorithm~\ref{algo:mine} in order to find left-minimal nodes in the suffix tree. If we consider the output restrictions as part of the problem, then we do not have to report minimally $(y,d)$-unique substrings that occur only once, i.e., leaves in ${\cal T}(S^\rev)$. Then, for all nodes $u$ such that all children of $u$ are leaves, we can directly compute the highest value of $d<\ell$ such that $L(u)^\rev$ is $(y,d)$-unique. This leads to the definition of the $\mbox{{\em fast}}\_h(u,\ell)$ function for a node $u$ of ${\cal T}(S^\rev)$.
The function $\mbox{{\em fast}}\_h(u,\ell)$ is defined similarly to the function $h(u,\ell)$ with the additional following case:
\begin{itemize}
	\item If all children of $u$ are leaves, we can directly compute the highest value of $d < \ell$ such that $L(u)^\rev$ is $(y,\ell)$-unique as the largest value $d < \ell$ such that, for each child $v$ of $u$,  $h(v,d+1) = d$. In other words, we are looking for the largest $d < \ell$ such that all children of $v$ are $(y,d)$-unique.
\end{itemize}
The definition of $\mbox{{\em fast}}\_h(u,\ell)$ is as follows:

$$ \mbox{{\em fast}}\_h(u,\ell) =
\begin{cases}
	\ell -1 & \text{ if } u \text{ is a leaf and } \sn(u) < \ell,\\
	\max\{ i < \ell \mid f^\rev(\sn(u)-i) = y \} & \text{ if } u \text{ is a leaf and such } i \text{ exists, } \\
	\max\{ i < \ell \mid \mbox{{\em fast}}\_h(v,i+1) = i &\text{ if all children of } u \text{ are leaves}\\
	\phantom{\max\{ i < \ell \mid f}\text{ for all } v \text{ child of } u \} & \phantom{\text{ if all childre}}\text{and such } i \text{ exists, }\\
	\min\{ \mbox{{\em fast}}\_h(v,\ell) \mid v \text{ child of } u \} & \text{ if } u \text{ is an inner node,} \\ -1 & \text{ otherwise. }
\end{cases}$$

The additional case of $\mbox{{\em fast}}\_h(u,\ell)$ can be computed as follows. Let $u$ be a node such that all children of $u$ are leaves. We set $i = \ell$, and compute the values $h(v,i)$ where $v$ is a child of $u$. We update the value of $i = \min(i,h(v,i)+1)$, compute the value of $h(v',i)$ where $v'$ is the next child of $u$, and update the value of $i = \min(i,h(v',i)+1)$. We continue iterating until all children $v$ of $u$ have the same value $h(v,i)$, possibly $-1$. Algorithm~\ref{algo:lazy_boost} summarizes these improvements to Algorithm~\ref{algo:lazy}.
In order to use the $\mbox{{\em fast}}\_h(u,\ell)$ function in Algorithm~\ref{algo:mine}, it is enough to replace the $h()$ function at line~\ref{line:h'} by the $\mbox{{\em fast}}\_h()$ function.

\IncMargin{1em}
\begin{algorithm}[tbp]
  \SetKwInOut{Input}{input}\SetKwInOut{Output}{output}
  \SetKwProg{Procedure}{procedure}{:}{end}
  \DontPrintSemicolon
  \Input{A node $u$ in the suffix tree ${\cal T}(S^\rev)$, and integer $\ell$.}
  \Output{Maximum delay $d<\ell$ such that $L(u)^\rev$ can be $(y,d)$-unique.}
  \LinesNumbered
  \BlankLine
  \Procedure{$fast\_h(u,\ell)$}{
    $ min_d \gets \ell -1$\;
    \If{$u$ is a leaf}{
      $ j \gets n - \sn(u) + 1$\;
      $min_d \gets \max\{{\tt select}(b_y, {\tt rank}(b_y, j+\ell)) -j, -1\}	$\;
    }
    \ElseIf{all children of $u$ are leaves}{
        \Repeat{$is\_changed$ AND $min_d \geq 0$}{
        $ is\_changed \gets false$\;
          \ForAll{children $v$ of $u$}{
            $ d = fast\_h(v,min_d+1)$\;
            \If{$min_d < d$}{
              $ is\_changed \gets true$\;
              $ min_d \gets d$\;
            }
          }
        }
      }
    \Else{
      \ForAll{children $v$ of $u$}{
        $ d = fast\_h(v,\ell)$\;
        \If{$min_d < d$}{
        $ min_d \gets d$\;
        }
      }
    }
    $ IPQ.${\tt demote($u,min_d$)}\;
    \Return $min_d$\;
  }

  \caption{{\sc Highest possible value of $d$.}}\label{algo:lazy_boost}
\end{algorithm}
\DecMargin{1em}

\section{Experimental results}\label{sec:experiments}

We implemented the algorithms presented in the previous sections and measured their performance on randomly generated datasets and on real-world datasets. The implementation is available online at %
\url{https://github.com/maxrossi91/colored-strings-miner}.

\subsection{Setup}

Experiments were performed on a 3.4\,GHz Intel Core i7-6700 CPU equipped with
8\,MiB L3 cache and 16\,GiB of DDR4 main memory. The machine had no other significant
CPU tasks running, and only a single thread of execution was used.

The OS was Linux (Ubuntu 16.04, 64bit) running kernel 4.4.0. All programs were
compiled using {\tt g++} version 5.4.0 with {\tt-O3} {\tt-DNDEBUG} {\tt -funroll-loops} {\tt -msse4.2} options.
All given runtimes were recorded with the C++11 {\tt high\_resolution\_clock} time measurement facility.

\subsection{Data}

We used two different datasets; the first one consists of randomly generated data, while the second one consists of real-world data.

The randomly generated data are colored strings generated using the C library function {\tt rand()}. We varied the length $n = \num{100},\num{1000},\num{10000},\num{100000}$, the alphabet size $\sigma = \num{2},\num{4},\num{8},\num{16},\num{32}$, and the number of colors $\gamma = \num{2},\num{4},\num{8},\num{16},\num{32}$. In all cases except for $n = \num{100000}$, we used seeds $\num{0}$, $\num{9843}$, $\num{27837}$, $\num{19341}$, $\num{29044} $; for $n = \num{100000}$,  we used only seed $\num{0}$. The string is generated one character (and its color) at a time, i.e.\ fixing $\sigma$ and $\gamma$, the string of length $n = \num{1000}$ is a prefix of the string $n=\num{10000}$. The strings are generated using a uniform distribution of characters and colors.
We report only the results of experiments for the values of length $n = \num{1000},\num{10000},\num{100000}$, alphabet size $\sigma = \num{2},\num{8},\num{32}$, number of colors $\gamma = \num{2},\num{8},\num{32}$, and seed $\num{0}$, since these are representative of the trend we observed in all our experiments.

The real-world data is the result of a simulation on a set of established benchmarks in embedded systems verification~\cite{brglez1989combinational,corno2000rt,OpenCore}, reported in Table~\ref{Table:real data sets}. The benchmarks are descriptions of hardware design at the register-transfer level (RTL) of abstraction. Each design consists of a set of primary input bits (PIs) and a set of primary output bits (POs). Primary inputs and primary outputs are grouped into ports. The simulation of designs is a sequence of temporal events which act to capture the effects of the values given as inputs for the design into the design itself, and consequently the effects of the input values on the values assumed by the outputs. We simulated the benchmarks providing as inputs randomly generated sequences using an automatic test pattern generator (ATPG). The result of the simulation is collected in a simulation trace, which stores, for each temporal event, the values of the primary inputs and of the primary outputs. For each simulation event, we consider the values of all primary inputs as characters of the alphabet $\Sigma$, and the values of a port of the primary outputs as colors. In other words, for simulation event $i$, $S[i]$ is the value of the primary inputs, and $f_S(i)$ is the value of the primary outputs.

\begin{table}[tbp]
	\centering
	\begin{tabular}{llrrrrrr}
		\hline
		Design             & Description & PIs & POs & $n$ & $\sigma$ & $\gamma$ & $n_y$  \\
		\hline
		{\tt b03}             & Resource arbiter~\cite{corno2000rt} & 6 & 4 & \num{100000}   & 17 & 5 & 3210 \\
		{\tt b06}		 &  Interrupt handler~\cite{corno2000rt} & 4 & 6 & \num{100000} & 5 & 4 & \num{44259} \\
		{\tt s386}           & Shynthetized controller~\cite{brglez1989combinational} & 9 & 7 & \num{100000} & 129 & 2 & 8290\\
		{\tt camellia}		     & Symmetric key block cypher~\cite{OpenCore} & 262 & 131 & \num{103615}  & 70 & 224 & 2292 \\
		{\tt serial}		     & Serial data transmitter  & 11 & 2 & \num{100000}  & 118 & 2 & \num{16353}  \\
		{\tt master}		     & Wishbone bus master~\cite{OpenCore}  & 134 & 135 & \num{100000}  & 417 & 80 & 759 \\
		\hline
	\end{tabular}
	\caption{Real-world datasets used in the experiments. In the column {\em Design} and {\em Description} we report the name and the description of the hardware design that we used to generate the simulation trace. In column {\em PIs} we give the number of primary inputs of the design, while in {\em POs} that of its primary outputs. In column $n$ we report the length of the simulation trace, and in  columns $\sigma$ and $\gamma$ the size of the alphabet and the number of colors, respectively. For each design we fixed a color $y$, and the value $n_y$ refers to the number of $y$ characters in the simulation trace.\label{Table:real data sets}}
\end{table}

\subsection{Algorithms}

We compared the following implementations:

\begin{itemize}
	\item {\bf\tt base}: the baseline algorithm (Algorithm~\ref{algo:algo2})
	\item {\bf\tt skip}: the skipping algorithm (Algorithm~\ref{algo:mine}) using the $h$ function (Algorithm~\ref{algo:lazy})
	\item {\bf\tt real}: the skipping algorithm (Algorithm~\ref{algo:mine}) using the $fast\_h$ function (Algorithm~\ref{algo:lazy_boost})
	\item {\bf\tt base-all}: the baseline algorithm for all colors (Algorithm~\ref{algo:algo3})
\end{itemize}

All algorithms report minimally $(y,d)$-unique substrings only if they are {\em real type}. We used the {\tt sdsl-lite} library~\cite{gbmp2014sea} for compressed suffix trees, range maximum query, and rank and select supports for bit vector implementations.

\subsection{Results}

We performed all experiments five times and report the average execution time over the five runs. Experimental results are reported in Figures~\ref{fig:random_results} and~\ref{fig:real_results}, and Table~\ref{tab:properties}.

\begin{figure}[tbp]
    \centering
 	\includegraphics[width=\textwidth]{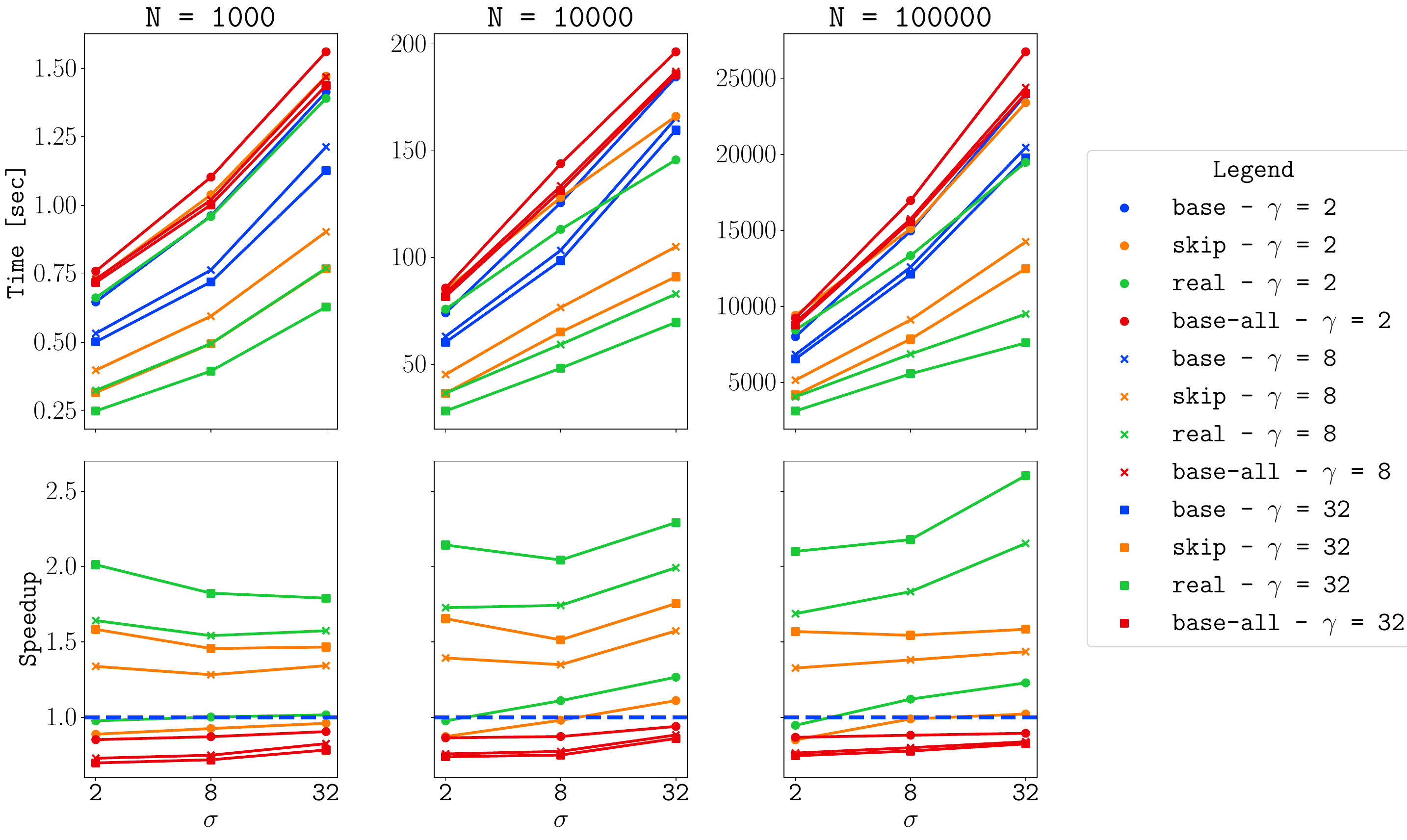}
    \caption{Results of the execution of algorithms {\tt base} (color blue), {\tt skip} (color orange), {\tt real} (color green), and {\tt base-all} (color red) over the randomly generated data for $N = 10^3$, $10^4$, and $10^5$. The $x$ axis represents the values of $\sigma = \{2,8,32\}$, and the different markers represents the values of $\gamma = \{2,8,32\}$ (circles, crosses, and boxes, respectively). The three plots in the first row report execution times. The plots in the second row report speedups of {\tt skip}, {\tt real}, and {\tt base-all} with respect to algorithm {\tt base} represented as the dashed blue line at constant $1.0$.\label{fig:random_results}\label{fig:all_colors_random}}
\end{figure}

\begin{figure}[tbp]
    \centering
 	\begin{subfigure}[b]{.49\textwidth}
 		\includegraphics[width=\textwidth]{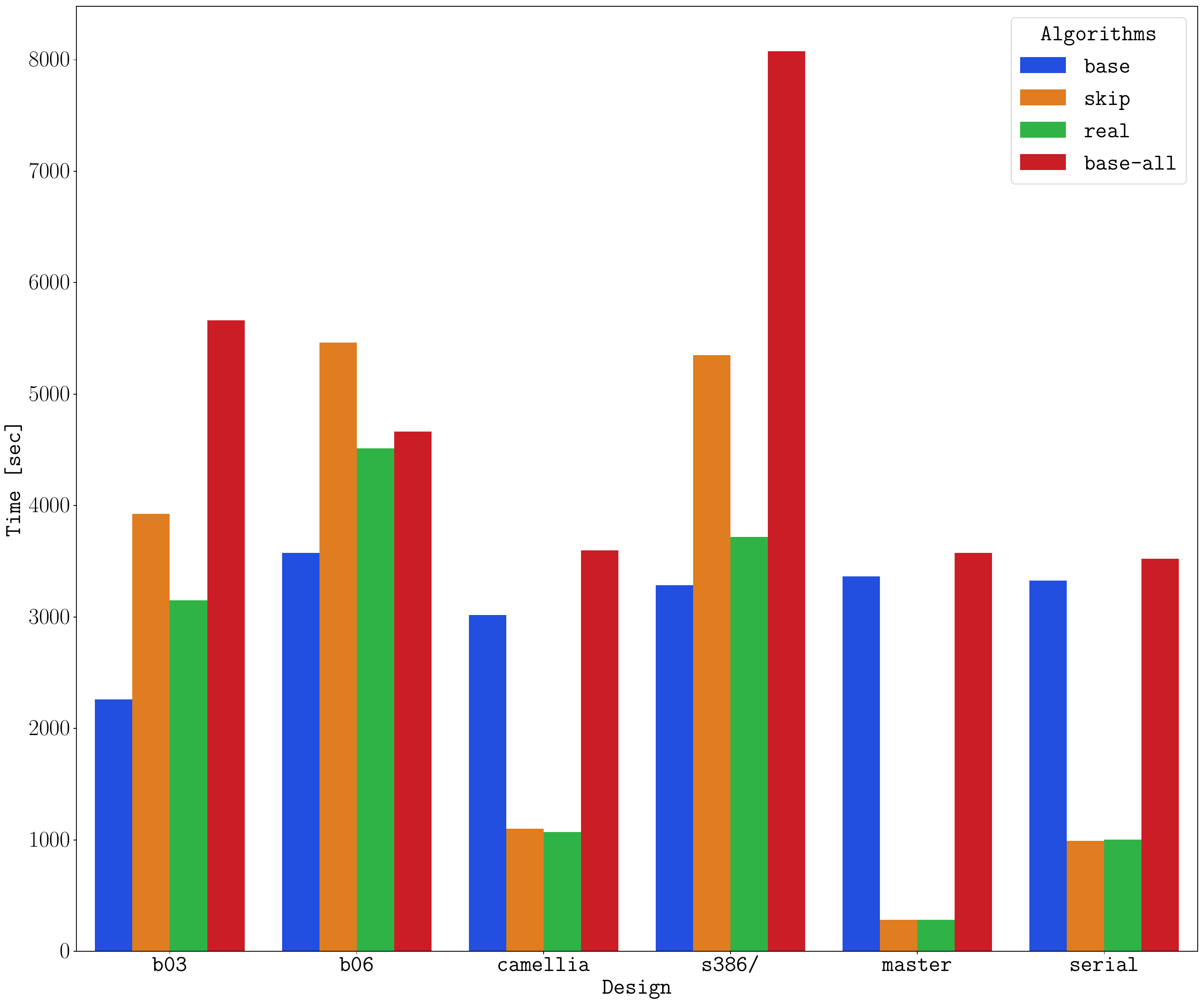}
 		\caption{Time\label{fig:real:time}}
 	\end{subfigure}
 	\hfill
 	\begin{subfigure}[b]{.49\textwidth}
 		\includegraphics[width=\textwidth]{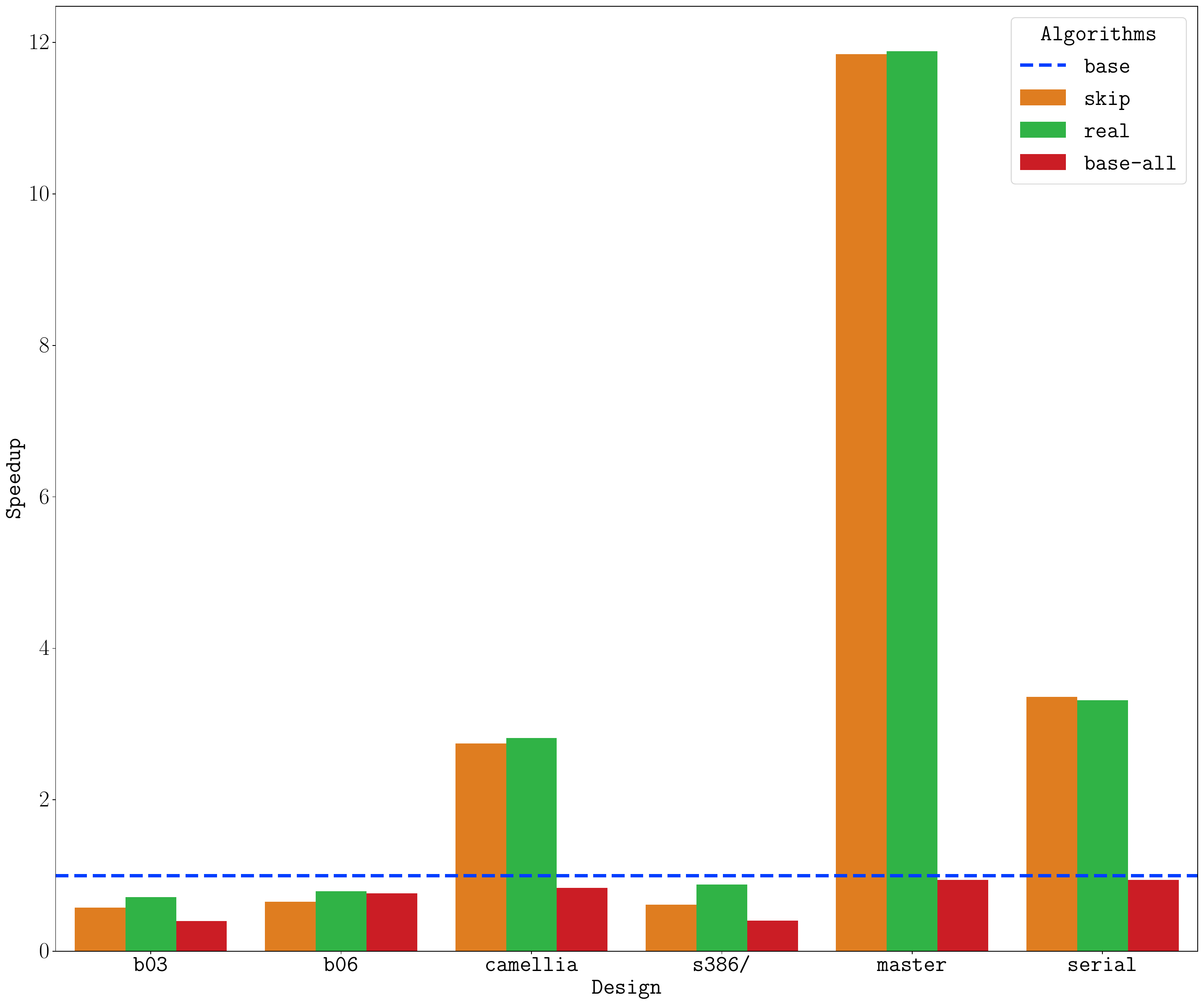}
 		\caption{Speedup\label{fig:real:speedup}}
 	\end{subfigure}
    \caption{Results of the execution of algorithms {\tt base}, {\tt skip}, {\tt real}, and {\tt base-all} over the real-world dataset. The plot in Figure~\ref{fig:real:time} reports execution times. The plot in Figure~\ref{fig:real:speedup} reports speedups of {\tt skip}, {\tt real}, and {\tt base-all} with respect to algorithm {\tt base} represented as the dashed blue line at constant $1.0$.\label{fig:real_results}\label{fig:all_colors_real}}
\end{figure}

\paragraph{Single color analysis}
Figure~\ref{fig:random_results} shows the results of the executions of {\tt base}, {\tt skip}, and {\tt real} algorithms over the randomly generated strings data. 

We can observe how the algorithms scale

\begin{enumerate}
	\item with respect to an increase in the numbers of colors, which has the effect of reducing the number of $y$-colored characters;
	\item with respect to an increase in the alphabet size; and
	\item with respect to an increase in the length of the text.
\end{enumerate}

We see that all three algorithms behave the same in all cases. Increasing the number of colors (case 1), the running time decreases. Conversely, when the size of the text alphabet increases (case 2), the running time increases also. Finally, we observe a quadratic dependence of the running time on text length (case 3); this is in accordance with our theoretic results (see Sec.~\ref{sec:baseline} and~\ref{sec:skipping}).

Figure~\ref{fig:random_results} shows that the {\tt skip} algorithm is almost always faster than the {\tt base} algorithm, and that the average speedup is \num{1.30}, with a maximum of  \num{1.75}. Moreover, we have that the the {\tt real} algorithm is almost always faster than the {\tt skip} algorithm, and the average speedup is \num{1.25}, with a maximum of  \num{1.64}. Finally, the average speedup between {\tt real} and {\tt base} is \num{1.65}, with a maximum of  \num{2.60} in the case of $N=\num{100000}$, $\sigma = 32$ and $\gamma = 32$.

 Figure~\ref{fig:real_results} shows the results for {\tt base}, {\tt skip} and {\tt real} algorithms on the real-world dataset. 
 Here, we observe a similar trend to the random data, but the speedup of {\tt real} with respect to {\tt base} is much higher --- \num{3.40} on average, with a maximum of \num{11.88} on the {\tt master} device. However, on three of the six datasets, {\tt base} is faster than {\tt skip}, and faster than {\tt real}.

\paragraph{All colors analysis}
Next, we compare the experimental results of the {\tt base} and {\tt base-all} algorithms. The setup is the same as in the previous case, i.e.\ we performed five runs of each experiment and give the average execution time. The results on the randomly generated and real-world datasets are reported in Figure~\ref{fig:all_colors_random} and~\ref{fig:all_colors_real}, respectively.

From Figures~\ref{fig:all_colors_random} and~\ref{fig:all_colors_real}, we can observe that the running time of the {\tt base} algorithm is not heavily affected by the fact that we are looking for a specific color: there is a small increase in running time from {\tt base} (reporting all real-type minimally $(y,d)$-unique substrings for any $d$ and just {\em one} color $y$), to {\tt base-all}, reporting all real-type minimally $(y,d)$-unique substrings for any $d$ and for {\em all} colors $y$.

On the real-world data, {\tt base} outperforms {\tt base-all} on all six datasets. Note that the number of patterns is considerably larger for {\tt base-all}. 

\begin{table}[t!]
	\begin{subtable}{.49\textwidth}
		\centering
        \begin{tabular}{rrrlrrr}
            \toprule
            \multicolumn{3}{r}{Alphabets}                                      &  & \multicolumn{2}{c}{Number of Properties}                         \\
            & {$\sigma$}                    & {$\gamma$}                    &  & {{\tt base}}                 & {\tt base-all}              \\ \hline
            & \num{2}                    & \num{2}                    & & \num{26894}                & \num{50922}                 \\
            &                      & \num{8}                    &  & \num{1745}                 & \num{15563}                 \\
            &                      & \num{32}                   &  & \num{76}                   & \num{3996}                  \\\\[-1em]
            & \num{8}                    & \num{2}                    &  & \num{30219}                & \num{56241}                 \\
            &                      & \num{8}                    &  & \num{1516}                 & \num{12919}                 \\
            &                      & \num{32}                   &  & \num{75}                   & \num{3306}                  \\\\[-1em]
            & \num{32}                   & \num{2}                    &  & \num{25120}                & \num{46758}                 \\
            &                      & \num{8}                    &  & \num{1245}                 & \num{10578}                 \\
            &                      & \num{32}                   &  & \num{40}                   & \num{2585}                  \\
            \bottomrule
        \end{tabular}
		\caption{Number of properties for $N=$\num{1000}.  \label{tab:all_colors_random1000}}
	\end{subtable}
	\begin{subtable}{.49\textwidth}
		\centering
        \begin{tabular}{rrrlrrr}
            \toprule
            \multicolumn{3}{r}{Alphabets}                                      &  & \multicolumn{2}{c}{Number of Properties}                         \\
            & {$\sigma$}                    & {$\gamma$}                    &  & {{\tt base}}                 & {\tt base-all}              \\ \hline
            & \num{2}                    & \num{2}                    &  & \num{2374231}              & \num{4699647}               \\
            &                      & \num{8}                    &  & \num{187202}               & \num{1447913}               \\
            &                      & \num{32}                   &  & \num{11767}                & \num{370303}                \\\\[-1em]
            & \num{8}                    & \num{2}                    &  & \num{2844680}              & \num{5607007}               \\
            &                      & \num{8}                    &  & \num{167431}               & \num{1294765}               \\
            &                      & \num{32}                   &  & \num{9989}                 & \num{317444}                \\\\[-1em]
            & \num{32}                   & \num{2}                    &  & \num{1466242}              & \num{2892791}               \\
            &                      & \num{8}                    &  & \num{83320}                & \num{642806}                \\
            &                      & \num{32}                   &  & \num{4948}                 & \num{156947}                \\
            \bottomrule
        \end{tabular}
		\caption{Number of properties for $N=$\num{10000}.  \label{tab:all_colors_random10000}}
	\end{subtable}
    \begin{subtable}{.49\textwidth}
        \centering
        \begin{tabular}{rrrlrrr}
            \toprule
            \multicolumn{3}{r}{Alphabets}                                      &  & \multicolumn{2}{c}{Number of Properties}                         \\
            & {$\sigma$}                    & {$\gamma$}                    &  & {{\tt base}}                 & {\tt base-all}              \\ \hline
            & \num{2}                    & \num{2}                    &  & \num{239039415}            & \num{473572454}             \\
            &                      & \num{8}                    &  & \num{17680770}             & \num{145246888}             \\
            &                      & \num{32}                   &  & \num{1129991}              & \num{37254203}              \\\\[-1em]
            & \num{8}                    & \num{2}                    &  & \num{279720849}            & \num{552304418}             \\
            &                      & \num{8}                    &  & \num{15517256}             & \num{127614675}             \\
            &                      & \num{32}                   &  & \num{947858}               & \num{31264100}              \\\\[-1em]
            & \num{32}                   & \num{2}                    &  & \num{243283926}            & \num{479770368}             \\
            &                      & \num{8}                    &  & \num{11982556}             & \num{98601898}              \\
            &                      & \num{32}                   &  & \num{713137}               & \num{23592691}              \\
            \bottomrule
        \end{tabular}
		\caption{Number of properties for $N=$\num{100000}.  \label{tab:all_colors_random100000}}
	\end{subtable}
	\begin{subtable}{.49\textwidth}
		\centering
		\begin{tabular}{llrr}
			\toprule
			&  & \multicolumn{2}{c}{Number of Properties}            \\
			{Design}       &  & {\tt base}        & {\tt base-all}       \\ \hline
			{\tt b03}      &  & \num{999191}            & \num{361224140}            \\
			{\tt b06}      &  & \num{223070824}         & \num{409476680}            \\
			{\tt s386}     &  & \num{5012263}           & \num{558001254}            \\
			{\tt camellia} &  & \num{77261}             & \num{2470894}              \\
			{\tt serial}   &  & \num{2085855}           & \num{11653080}             \\
			{\tt master}   &  & \num{252231}            & \num{34812555}             \\  \bottomrule
		\end{tabular}
		\caption{Number of properties for real-world dataset. \label{tab:all_colors_real}}
	\end{subtable}
	\caption{Number of properties for randomly generated strings and the real-world dataset. The first two columns of~\ref{tab:all_colors_random1000},~\ref{tab:all_colors_random1000}, and ~\ref{tab:all_colors_random1000} report the size of the text alphabet $\sigma$ and the number of colors $\gamma$, while the first column of~\ref{tab:all_colors_real} reports the name of the design where the simulation trace is retrieved. The last two columns report the number of properties extracted from the {\tt base} and the {\tt base-all} algorithms, respectively.\label{tab:properties}}
\end{table}

\section{Conclusion}\label{sec:conclusion}

We studied pattern discovery problems on colored strings
motivated by applications in embedded system verification.
To the best of our knowledge this is the first principled algorithmic treatment of these problems. 

Colored strings are strings such that each position of the string is assigned a color from a finite set of colors. We studied two different pattern discovery problems on colored strings. The first problem is to find all minimally $(y,d)$-unique substrings of the colored string, for a given color $y$ and any delay $d$. We proposed two different approaches, which we refer to as baseline approach and skipping approach. Both algorithms use a suffix tree on the reverse of the colored string as underlying data structure. They discover the patterns starting from the ones with highest delay to the ones with the lowest delay. The two algorithms differ in the way in which minimality information is propagated along the suffix tree. The baseline algorithm traverses the whole tree separately for each delay value, propagating a coloring function from the leaves to the root of the suffix tree. During each traversal, the algorithm goes through all distinct substrings of the text, and uses the coloring function to identify which substrings are minimally $(y,d)$-unique.

On the other hand, the skipping algorithm stores, for each distinct substring, the next delay value such that the substring is $(y,d)$-unique, during the discovery process. It uses a maximum-oriented indexed priority queue to find these values and to identify minimally $(y,d)$-unique substrings. 

Even though the theoretical analysis we provided for the skipping algorithms results in a worse upper bound on the running time than for the baseline algorithm, we show in our experiments that it is faster in practice on simulated data, and on half of the real-world data. %
Even though the skipping algorithm outperforms the baseline algorithm only on half of the real-world datasets (camelia, master, serial), it is significant that the gain on these is considerable, as opposed to only a slight slowdown on the others (b03, b06, s386), see Figure~\ref{fig:real_results}. Moreover, the traces on which we see a speedup are derived from devices which perform more complex tasks, thus indicating that our algorithm may be well suited for the application in  embedded system verifiaction. 

We also proposed a variant of the minimality condition oriented toward real-world application instances. Those conditions allow us to develop a faster core function of the skipping algorithm, resulting in a more effective performance in practice.

The second problem we proposed is to find all minimally $(y,d)$-unique substrings of the colored string, for {\em all} colors $y$ concurrently. We modified our baseline algorithm, defining a new coloring function, noting that for fixed $d$, a substring can be $(y,d)$-unique for at most one color $y$. The introduction of the new coloring function and the fact that now all colors $y$ are of interest, increases the running time with respect to the baseline algorithm only negligibly, an effect we observed in the experiments on both randomly generated data and real-world data.

We are currently working with colleagues in embedded systems to integrate these algorithms into their analysis workflows.

\section*{Acknowledgements}
  We thank Johannes Fischer, Travis Gagie, and Ferdinando Cicalese for interesting discussions, and Alessandro Danese for providing an updated data set of traces. We also thank the reviewers for carefully reading the paper. SJP's research was partially funded by the Academy of Finland via grant 319454. MR was partially funded by the Scuola di Dottorato dell'Universit\`a degli Studi di Verona, Italy, and is supported by the National Science Foundation (NSF) IIS (Grant No. 1618814).

\bibliography{colored}

\end{document}